\documentclass[a4paper,UKenglish,cleveref, autoref, thm-restate]{lipics-v2021}
\usepackage[ruled,vlined]{algorithm2e}
\usepackage{xcolor}
\usepackage[normalem]{ulem}
\usepackage{thm-restate}
\usepackage{multicol}
\usepackage{multirow}
\usepackage{tabularray}
\usepackage[justification=centering]{caption}

\pdfoutput=1 

\hideLIPIcs  

\graphicspath{{./graphics/}}

\title{Practical algorithms for Hierarchical overlap graphs}

\author{Saumya Talera}{Department of Computer Science and Engineering, Indian Institute of Technology Roorkee, India}{saumya_t@cs.iitr.ac.in}{}{}

\author{Parth Bansal}{Department of Computer Science and Engineering, Indian Institute of Technology Roorkee, India}{parth_b@cs.iitr.ac.in}{}{}

\author{Shabnam Khan}{Department of Computer Science and Engineering, Indian Institute of Technology Roorkee, India}{shabnam_k@cs.iitr.ac.in}{}{}

\author{Shahbaz Khan}{Department of Computer Science and Engineering, Indian Institute of Technology Roorkee, India}{shahbaz.khan@cs.iitr.ac.in}{https://orcid.org/0000-0001-9352-0088}{}

\authorrunning{S. Talera, P. Bansal, S. Khan, S. Khan} 

\Copyright{Saumya Talera, Parth Bansal, Shabnam Khan, and Shahbaz Khan} 



\category{} 

\relatedversion{} 

\supplement{https://github.com/renauga/hog}

\funding{Startup Research Grant SRG/2022/000801 by DST SERB, Government of India. }

\nolinenumbers 

\EventEditors{John Q. Open and Joan R. Access}
\EventNoEds{2}
\EventLongTitle{42nd Conference on Very Important Topics (CVIT 2016)}
\EventShortTitle{CVIT 2016}
\EventAcronym{CVIT}
\EventYear{2016}
\EventDate{December 24--27, 2016}
\EventLocation{Little Whinging, United Kingdom}
\EventLogo{}
\SeriesVolume{42}
\ArticleNo{23}

\newcommand{\sbzkc}[1]{{\color{cyan} (#1)}}
\newcommand{\sbzkr}[1]{{\color{red} \sout{#1}}}

\begin{document}

\maketitle

\begin{abstract}
One of the most prominent problems studied in bioinformatics is genome assembly, where, given a set of overlapping substrings of a source string, the aim is to compute the source string. Most classical approaches to genome assembly use assembly graphs built using this set of substrings to compute the source string efficiently. Prominent such graphs present a tradeoff between scalability and avoiding information loss. The space efficient (hence scalable) de Bruijn graphs come at the price of losing crucial overlap information. On the other hand, complete overlap information is maintained by overlap graphs at the expense of quadratic space.
Hierarchical overlap graphs were introduced by Cazaux and Rivals~[IPL20] to overcome these limitations, i.e., avoiding information loss despite using linear space. However, their algorithm required superlinear space and time. After a series of suboptimal improvements, two optimal algorithms were simultaneously presented by Khan~[CPM2021] and Park~et~al.~[CPM2021].


We empirically analyze all the algorithms for computing Hierarchical overlap graphs, where the optimal algorithm~[CPM2021] outperforms the previous algorithms as expected. 
However, it is still based on relatively complex arguments for its formal proof and uses relatively complex data structures for its implementation. We present an  \textit{intuitive}, \textit{optimal} algorithm requiring linear space and time which uses only \textit{elementary arrays} for its implementation. Despite the superior performance of optimal algorithm~[CPM2021] over previous algorithms, it comes at the expense of extra memory. Our algorithm empirically proves even better for both time and memory over all the algorithms, highlighting its significance in both theory and practice. 


We further explore the applications of hierarchical overlap graphs to solve variants of suffix-prefix queries on a set of strings, recently studied by Loukides et al.~[CPM2023]. They presented state-of-the-art algorithms requiring complex black-box data structures, making them seemingly impractical. Our algorithms, despite failing to match their theoretical bounds, answer queries in $0.002$-$100~ms$ for datasets having around a billion characters.

\end{abstract}

\newpage

\section{Introduction}
\label{sec:introduction}
Genome assembly is one of the most fundamental problems in Bioinformatics. For a source string (\textit{genome}), practical limitations allow sequencing of only a collection of its substrings (\textit{reads}) instead of the whole string. Thus, genome assembly aims to reconstruct the source genome using the sequenced reads. This is possible only because the sequencing ensures coverage of the entire string so that we find sufficient overlap among the substrings. This overlap information plays a fundamental role in solving the genome assembly problem. In most practical approaches~\cite{Velvet,DBLP:journals/jcb/BankevichNAGDKLNPPPSVTAP12,DBLP:conf/recomb/NurkMKP16,DBLP:journals/bioinformatics/AntipovKMP16,PevznerTW01,DBLP:journals/bioinformatics/SimpsonD10}, this overlap information is efficiently processed by representing the substrings in the form of \textit{assembly graphs}, such as \textit{de Bruijn graphs}~\cite{Pevzner1989} and \textit{overlap graphs} (or string graphs~\cite{DBLP:conf/eccb/Myers05}). The de Bruijn graphs present a fundamental trade-off of achieving scalability (space efficiency) at the expense of loss in overlap information. This is possible by considering all possible substrings (\textit{k-mers}) of the reads having length $k$, and storing the limited overlap information amongst them, thereby ensuring linear space. On the other hand, overlap graphs store the maximum overlap between every pair of reads \textit{explicitly}, thereby requiring quadratic size, making them impractical for large data sets.  

Hierarchical Overlap Graphs (HOG) were formally introduced by Cazaux and Rivals~\cite{DBLP:journals/ipl/CazauxR20} as an alternative to overcome the limitations of the existing assembly graphs, maintaining the complete overlap information \textit{implicitly} using optimal space (linear). Structurally it is similar to AC Trie~\cite{DBLP:journals/cacm/AhoC75} where HOG has nodes corresponding to the original substrings and the maximum overlap between every pair of these substrings only (instead of all prefixes as in AC Trie). Notably, the linear size of HOG (despite storing complete overlap information) highlights the redundancy of data stored in overlap graphs. Several applications of HOG have been formerly studied in~\cite{DBLP:conf/dcc/CazauxCR16,DBLP:journals/corr/CanovasCR17}. For a given set $P$ of $k$ strings with total length $n$, they~\cite{DBLP:journals/ipl/CazauxR20} compute the HOG in $O(n+k^2)$ time using superlinear space. Later, Park et al.~\cite{DBLP:conf/spire/ParkCPR20} presented an algorithm requiring $O(n\log k)$ time and linear space. However, they assumed the character set of the strings to be of constant size. Finally, Khan~\cite{DBLP:conf/cpm/000421} and Park et al.~\cite{DBLP:conf/cpm/ParkPCPR21} independently presented optimal algorithms for computing the HOG in $O(n)$ time. While the former algorithm only used standard data structures (as lists, stacks), the latter reduced the problem to that of computing borders from the classical KMP algorithm~\cite{DBLP:journals/siamcomp/KnuthMP77}.\\ 

\noindent
\textbf{Related Work.}
All the above algorithms, in theory, remove the non-essential nodes from an AC Trie (ACT)~\cite{DBLP:journals/cacm/AhoC75}. In practice, they can use an intermediate \textit{Extended HOG} (EHOG) to reduce the memory requirement for these algorithms. Thus, the \textit{fundamental} problem is marking those nodes of an ACT (or EHOG), which represent the set of maximum overlaps between every ordered pair of strings in $P$ which is a special case of the All Pairs Suffix Prefix (APSP) problem (not to be confused with the classical graph problem of computing All Pairs Shortest Paths). APSP aims to calculate the maximum overlap between every pair of strings among a set of strings. Gusfield et al.~\cite{DBLP:journals/ipl/GusfieldLS92} optimally solved this classical problem in $O(n+k^2)$ time and $O(n)$ space using generalized suffix tree~\cite{DBLP:conf/focs/Weiner73}. Later, other optimal algorithms were presented using generalized suffix array~\cite{DBLP:journals/siamcomp/ManberM93} (\cite{DBLP:journals/ipl/OhlebuschG10,DBLP:journals/jda/TustumiGTL16}) and Aho Corasick automaton~\cite{DBLP:journals/cacm/AhoC75} (\cite{DBLP:journals/ipl/LoukidesP22}). Moreover, in practice, several algorithms having suboptimal bounds such as Readjoiner~\cite{DBLP:journals/bmcbi/GonnellaK12}, SOF~\cite{HajRachid2015}, and a recent algorithm by Lin and Park~\cite{DBLP:journals/tcs/LimP17}, have far superior performance for APSP, demonstrating the stark difference between theory and practice. Thus, any empirical analysis cannot trivially overlook the suboptimal algorithms. Note that, since the APSP problem \textit{explicitly} reports the maximum overlap for each of the $k^2$ ordered pair of strings, the above bound is optimal. However, when only the set of {\em maximum overlaps} is required (as in the case of HOG), the $O(k^2)$ factor proves suboptimal. 

Recently, Loukides et al.~\cite{DBLP:conf/cpm/LoukidesPTZ23} have studied several practically relevant variants of the APSP problem and presented state-of-the-art algorithms for the same. However, most of these algorithms involve complex black-box data structures, making them seemingly impractical. Another interesting variant of the APSP problem was solved by Ukkonen~\cite{DBLP:journals/algorithmica/Ukkonen90}, which reports all pairwise overlaps (not just maximum) in the decreasing order of lengths. Notably, this problem can be solved by EHOG by reporting the internal nodes bottom-up.




\subsection{Our Results}
Our results can be succinctly outlined as follows.

\begin{enumerate}
    \item \textbf{Novel optimal algorithm for HOG.}  We present arguably a very \textit{intuitive} algorithm, which is also \textit{optimal}. Our algorithm runs in linear time and space. Despite the simplicity of the existing optimal algorithms, they still use relatively complex arguments for their formal proof and relatively complex data structures for implementation. However, our algorithm uses only \textit{elementary arrays} for implementation with a very \textit{intuitive} proof.  

    \item \textbf{Empirical evaluation of practical HOG algorithms.} We analyzed all the seemingly practical algorithms on random and real datasets. On real datasets, the previous optimal algorithm improves the other algorithms by $1.3$-$8\times$ using $1.5$-$6\times$ more memory. However, the proposed optimal algorithm improves the best time by $1.3-1.7\times$ and memory by $\approx 1.5\times$. The relative performances improve with the size of the dataset. On random datasets having $~10^5$ strings with a total length of $10^7$, the optimal algorithm improves the previous state of the art by over $2-4\times$ times, while our algorithm improves the optimal by $2.5-3\times$. Further, we find that the performance mainly depends on an intermediately computed EHOG used by all the algorithms. Our results prove that the proposed algorithm performs superior to the previous algorithms on both random and real datasets, highlighting its significance in both theory and practice.
         
    \item \textbf{Applications of HOG.}
    We also explore the applications of HOG to solve various variants of suffix-prefix queries on a set of strings. Despite not matching the theoretical bounds of the state of the art~\cite{DBLP:conf/cpm/LoukidesPTZ23}, which require complex black-box data structures, our algorithms require $0.002$-$100~ms$  for queries on a data set of a billion characters, improving $18$-$1300\times$ over classical KMP~\cite{DBLP:journals/siamcomp/KnuthMP77} for complex queries.
\end{enumerate}

\noindent
\textbf{Outline of the paper.} We first describe the basic notations and definitions used throughout the paper in \Cref{sec:prelim}. In \Cref{sec:prevWork}, we briefly describe the previous results, highlighting the aspects affecting practical performance. This is followed by describing our proposed algorithm in \Cref{sec:algo}. Then, in \Cref{sec:exp}, we present the experimental setup and evaluation of the practical algorithms on real and random datasets, justifying the observations using further evaluation. \Cref{sec:apps} discussed some applications of HOG on variants of APSP queries. We finally present the conclusions and scope of future work in \Cref{sec:conclusion}.

\section{Preliminaries}
\label{sec:prelim}
We are given a set $P = \{p_1, p_2, ... p_k\}$ of $k$ non-empty strings over a finite-set alphabet. The length of a string $p_i$ is denoted by $|p_i|$. Let the total length of all strings $p_i\in P$ be $n$ ($\geq k$ as non-empty strings). An empty string is denoted by $\epsilon$. A substring of a string $p$ starting from the first character of $p$ is a \textit{prefix} of $p$, while a substring ending at the last character of $p$ is known as a \textit{suffix} of $p$. A prefix or suffix of $p$ is called \textit{proper} if it is not the same as $p$. For an ordered pair $(p,q)$, an $overlap$ is a string that is both a proper suffix of $p$ and a proper prefix of $q$, and $ov(p,q)$ denotes the maximum such overlap. We denote $Ov(P)$ as the set of $ov(p_i, p_j)$ for all $p_i, p_j \in P$. 
We now use the above definition to define the HOG.

For a given $P$, the ACT $\cal{A}$, EHOG $\cal{E}$, and HOG $\cal{H}$ are defined using different vertex sets as follows. The vertices of $\cal A$ are all possible prefixes of the strings in $P$ (hence including $P$ and $\epsilon$). The vertices of $\cal E$ are $P\cup\{\epsilon\}$ and all possible (not necessarily maximum) overlaps of strings in $P$. The vertices of $\cal H$ are further restricted to the maximum overlaps $Ov(P)\cup P\cup\{\epsilon\}$. Hence for a given $P$, it holds $V({\cal H})\subseteq V({\cal E})\subseteq V({\cal A})$. For all the above, we add the following types of edges: (a) \textit{Tree edges} $(y,x)$ for all $x\in V$ with label $z$, where $y$ is the longest proper prefix of $x$ in $V$ and $x=yz$, and (b) \textit{Suffix links} $(x,y)$ represented by dotted red edges, for all $x\in V$ where $y$ is the longest proper suffix of $x$ in $V$. See \Cref{fig:aehog} for examples of the same. Also, for each of the structures the number of nodes is $O(n)$. 


We abuse the notation to interchangeably refer to a node as a string (corresponding to it) and vice-versa. The name \textit{tree edges} allows us to abuse the notation and treat each of $\cal A, E$ and $\cal H$ as a tree with additional suffix links. This makes the outgoing tree edge neighbors of a node its children, which extends to the notion of descendants and ancestors. This also makes the \textit{leaves} of the tree a \textit{subset} of $P$, the empty string $\epsilon$ as the root, and the remaining strings as \textit{internal nodes}. Note that some string $p_i\in P$ can also be an internal node if it is a prefix of some other string $p_j\in P$, which is why leaves of the tree are not necessarily the same as $P$. Further, starting from every node $p_i\in P$ and following the suffix links until we reach the root gives us the \textit{suffix path} of $p_i$. We further have the following property.


\begin{figure}
    \centering
    \includegraphics[scale=1.2]{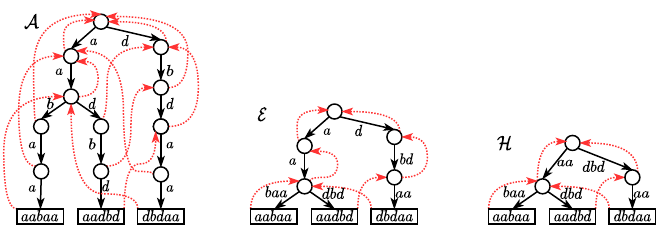}
    \caption{The structures ${\cal A}$, ${\cal E}$ and ${\cal H}$ for $P=\{aabaa,aadbd,dbdaa\}$ (reproduced from~\cite{DBLP:conf/cpm/000421}).}
    \label{fig:aehog}
\end{figure}

\begin{lemma}[HOG Property~\cite{DBLP:journals/ipl/GusfieldLS92}]
\label{lemma:inOv}
    An internal node $v$ in $\cal{A}$ of $P$, is $ov(p_i, p_j)$ for two strings $p_i, p_j \in P$ iff $v$ is an overlap of $(p_i, p_j)$ and no descendant of $v$ is an overlap of $(p_i, p_j)$.
\end{lemma}


\section{Previous work}
\label{sec:prevWork}
We now describe the relevant previous results related to the computation of HOG.
For a given set of strings $P$, the ACT~\cite{DBLP:journals/cacm/AhoC75} can be computed in linear time and space. The classical algorithm builds a Trie and adds the suffix links efficiently using dynamic programming. EHOG~\cite{DBLP:journals/ipl/CazauxR20} can be computed by simply traversing the suffix path of every string in $P$ in linear time, and removing the untraversed vertices. The corresponding labels of tree edges and suffix links with its labels are also updated simultaneously.



Given an ACT (or EHOG) of $P$, its HOG can be constructed by removing the nodes not satisfying \Cref{lemma:inOv}, and updating edge labels and suffix links as in EHOG. 
Various approaches to compute $Ov(P)$ are

\begin{enumerate}
    \item \textbf{Cazaux and Rivals~\cite{DBLP:journals/ipl/CazauxR20} requiring $O(n+k^2)$ time.} This algorithm computes a list $R_l(u)$ containing all leaves having ${u}$ as its suffix by traversing over all suffix paths. A bit vector maintains for each internal node $x$, the strings in $P$ having overlaps in the descendants of $x$. The bit vector can be updated in a bottom-up fashion, and compared with $R_l(u)$ to decide whether $u$ is added to $\cal H$. Along with suboptimal time, it also uses super linear space seemingly limiting its scalability. 

\item \textbf{Park et al.~\cite{DBLP:conf/spire/ParkCPR20} requiring $O(n\log k)$ time. }
  This algorithm sorts the strings in ${P}$ lexicographically, which makes a contiguous set of leaves the descendant of each internal node. This allows them to use the practical \textit{segment tree} data structure allowing update over an interval of in $O(\log k)$ time. For each node in the suffix path of $p_i\in P$, they simply count the number of descendant leaves (eligible prefixes) that are not covered by a larger prefix of $p_i$. This allows them to process each node of every suffix path (total $O(n)$) in a single segment tree query and update requiring linear space. 
 
\item \textbf{Khan~\cite{DBLP:conf/cpm/000421} requiring $O(n)$ time.}
This algorithm also computes $R_l(u)$ as in~\cite{DBLP:journals/ipl/CazauxR20} (referred to as ${\cal{L}}_u$) and performs a traversal of $\cal{A}$ which maintains the lowest internal node having each $p_i\in P$ as a suffix on top of a stack $S_i$. On reaching a leaf $p_j\in P$, each such internal node, say top of $S_i$, is the maximum overlap $ov(p_i,p_j)$. They only use stacks and lists for their implementation, giving a simple algorithm with a relatively complex proof. 

\item \textbf{Park et al.~\cite{DBLP:conf/cpm/ParkPCPR21} requiring $O(n)$ time.}
The algorithm computes the overlaps by reducing the problem to that of computing borders or failure links in the classical KMP algorithm~\cite{DBLP:journals/siamcomp/KnuthMP77}. This allows each internal node to uniquely identify the ancestors whose presence in $\cal H$ is affected by them (recall HOG property from \Cref{sec:prelim}) in linear time. 
\end{enumerate}

Note that, except for the last algorithm~\cite{DBLP:conf/cpm/ParkPCPR21}, all the algorithms can operate directly on $\cal E$ (instead of $\cal A$ or $P$) resulting in significant performance enhancement when $|\cal E|\ll |\cal A|$. 



\section{Proposed Algorithm}
\label{sec:algo}

As described earlier, all the algorithms essentially mark the nodes $Ov(P)$ of $\cal A$ (or $\cal E$) which will be added to $\cal H$, followed by a single traversal~\cite{DBLP:conf/cpm/000421} to generate $\cal H$. These algorithms take $\cal A$ (or $\cal E$) as input and produce $in{\cal H}$ array as an output. For any node $v$, $in{\cal H}[v]$ is a boolean such that if its value is true then $v$ should be in $\cal H$. We thus focus here on optimally computing this $in{\cal H}$ array. We describe our algorithm incrementally, \textit{firstly} to mark all overlaps for pairs $(p_i, p_j)$ for a fixed $p_i$ and all $p_j\in P$, \textit{secondly} optimize it to require $O(|p_i|)$ time after initialization, and \textit{finally} generalize it to find $Ov(P)$ in $O(n)$ time.
We use the HOG Property (\Cref{lemma:inOv}) to identify the nodes in $Ov(P)$.
For the sake of simplicity, we assume no string in $P$ is a prefix of another, making the nodes in $P$ be the leaves of $\cal A$. We shall later address how to handle this case separately.

\subsection{Marking all overlaps $ov(p_i,*)$ for a string $p_i$}
\label{sec:simpPro}
Recall that the suffix path of $p_i$ visits all the overlaps of $p_i$ in decreasing order of depth (moving closer to the root). Further, an internal node $x$ on this path can be an overlap of $(p_i,p_j)$ only if $p_j$ is a descendant of $x$ (as $x$ is a prefix of its descendants). Hence, for all the internal nodes which are overlaps for $(p_i,p_j)$, the node visited first would be $ov(p_i,p_j)$. 
Given the above properties, we propose the following simple process by colouring the leaves from \textit{white} to \textit{black} as follows. 

\textit{Starting with all leaves \textit{white}, for every node $x$ on the suffix path of $p_i$, we \textit{blacken} all the \textit{white} leaves in subtree rooted at $x$, marking $x$ if any leaf was \textit{blackened} while processing $x$. }

The correctness of the above algorithm is evident as $x$'s are visited bottom up along the suffix path, and only the leaves $p_j$ in subtree of $x$ can have $x$ as an overlap of $(p_i,p_j)$. 

\begin{algorithm}
\lForEach{node $v\in V({\cal A})$}{$in{\cal H}[v] \gets false$} 
$in{\cal H}[root] \gets true$\;
\lForEach{node $v\in V({\cal A})$}{$\mathit{count}[v]\gets$ Number of children of $v$ in tree ${\cal A}$}

\ForEach{node $x$ in  $P$}{
    $in{\cal H}[x] \gets true$\;
    $v \gets $ Suffix link of $x$\;
    $V_m \gets \emptyset$\;

    \While{$v \neq root$} {
        \If{$\mathit{count}[\mathit{favDesc}[v]] \neq 0$} {
            $in{\cal H}[v] \gets true$\;
            $u\gets \mathit{favDesc}[v]$\;
            \While{$u\neq root$} {
               Add $u$ to $V_m$\;
               \uIf{$u\neq \mathit{favDesc}[v]$}{
               $\mathit{count}[u]\gets \mathit{count}[u]-1$\;
               \lIf{$\mathit{count}[u]>0$}{\textit{break}}
               }\lElse{$\mathit{count}[u] \gets 0$}
                $u \gets \mathit{favPAnc}[u]$\;
            }
        }
        $v \gets$ Suffix link of $v$\;
    }
    \lForEach{node $v\in V_m$}{$\mathit{count}[v] \gets$ Number of children of $v$ in tree $\cal A$}
}
\textbf{return} $in{\cal H}$\;
   
\caption{$Compute{\cal H}({\cal A})$}
\label{alg:MarkH}
\end{algorithm}

\subsection{Optimally marking $ov(p_i,*)$ for a string $p_i$}
Clearly, visiting all descendants repeatedly is inefficient. The main idea behind the optimization is to lazily \textit{blacken} the entire subtree of each $x$ in \textit{constant} time. This is possible since future updates and queries about \textit{blackened} vertices (due to some suffix $x'$ of $x$) will only occur at vertices closer to the root than $x$, avoiding direct queries to $x$'s descendants. 

However, to limit the number of updates to \textit{constant}, we face another issue due to paths having nodes with single child (see \Cref{fig:aehog}). We circumvent it by maintaining such an entire path using its \textit{favoured} descendant, where we refer to nodes having multiple children as \textit{favoured}. Note a \textit{favoured} vertex is its own favoured descendant. The entire path having the same favoured descendant can be represented together because the number of white leaves of every node in such a path would always be the same. 

We thus maintain the following data to use the above approach efficiently. 
\begin{itemize}
    \item $\mathit{count}[x]:$ Number of child subtrees of $x$ having white leaves.
    \item $\mathit{favPAnc}[x]:$ The closest \textit{proper} ancestor of $x$ having multiple children. 
    \item $\mathit{favDesc}[x]:$ The closest descendant of $x$ (not necessarily proper) having multiple children.
    \item $V_m:$ List of vertices whose $\mathit{count}[\cdot]$ is modified while visiting a suffix path.
\end{itemize}

\noindent
\textbf{Initialization.} 
Initially, since all the leaves are \textit{white}, we initialize $V_m=\emptyset$ and $\mathit{count}[x]$ as the number of children of $x$ in the tree $\cal A$. Clearly, using a single traversal over the tree ${\cal A}$, we can initialize $\mathit{count}[\cdot]$, $\mathit{favPAnc}[\cdot]$ and $\mathit{favDesc}[\cdot]$ in $O(n)$ time.

\noindent
\textbf{Processing Suffix path of $p_i$.} 
Now, starting from the suffix of $p_i$ in ${\cal A}$, we process each node $x$ on the suffix path in order. The process of \textit{blackening} all the leaves of the subtree of $x$ can be performed \textit{lazily} because the future nodes on the suffix path of $p_i$ cannot be the descendants of $x$. Assuming the data structure is correctly maintained, the $x$ is marked in $in{\cal H}$ if $\mathit{count}[x]\neq 0$, implying that leaves can be \textit{blackened} while processing $x$ (recall the simple process in \Cref{sec:simpPro}). To \textit{blacken} all the leaves in the subtree of $x$ we set the count of its representative (nodes having a single child in a path are equivalent) $\mathit{favDesc}[x]$ to \textit{zero}. 
Recall that the suffix path of $p_i$ is processed in the decreasing order of depth, hence we don't need to update the count of any other descendant of $x$. Further, the entire data structure can be updated by updating the count of the ancestors of $x$ whose $\mathit{count}$ is affected. This includes all the ancestors of $y$ having a subtree containing white leaves only among the descendants of $x$. These ancestors can be accessed by repeatedly reducing counts of \textit{favoured} proper ancestors until some ancestor has other white leaves remaining in its other subtrees. 

\noindent
\textbf{Analysis.} For each $x$ on the suffix path of $p_i\in P$, except for repeatedly updating the $\mathit{count}$ of the favoured proper ancestors, all the other operations require \textit{constant} time.  We may have to process multiple such ancestors $y$ whose $\mathit{count}$ became zero, and at most one such ancestor $y^*$ whose $\mathit{count}$ remained non-zero. We thus use an amortized analysis argument by associating twice the \textit{cost} of processing the ancestor $y^*$, which pays off for its future processing when $\mathit{count}[y^*]$ will become \textit{zero}. Since every \textit{favoured} proper ancestor $y$ initially has $\mathit{count}[y]\geq 2$, every $y$ whose $\mathit{count}[y]$ became zero was once the unique $y'^*$ while processing $x'$ on the suffix path which made $\mathit{count}[y'^*]=\mathit{count}[y]=1$, and hence paid in advance for current \textit{cost}. Thus, for each $x$ along the suffix path, we process an amortized $O(1)$ favoured proper ancestors in $O(1)$ time. Hence, computing all $ov(p_i,*)$ for a $p_i\in P$ require $O(|p_i|)$ time after initialization of data structures in $O(n)$. 

\subsection{Final Algorithm and Complexity}
The final algorithm requires repeating the above algorithm for each $p_i\in P$. Note that both $\mathit{favPAnc}[\cdot]$ and $\mathit{favDesc}[\cdot]$ are not updated throughout the algorithm. In order to initialize $\mathit{count}[\cdot]$ efficiently,  we simply maintain a list $V_m$ of all the nodes with modified $\mathit{count}[\cdot]$ for each $x$ on the suffix path, which can be used to restore $\mathit{count}[\cdot]$ efficiently.
The time taken after initialization by an iteration processing $p_i$ is $O(|p_i|)$, we also have $V_m=O(|p_i|)$ and hence the data structures can be restored in $O(|p_i|)$. Thus, after initialization in $O(n)$ time, the total time required by the algorithm is $\sum_{p_i\in P} |p_i|=O(n)$, which is optimal.

\textbf{Remark:} The case of a string $p_i\in P$ being a prefix of another in $P$ can be handled as follows. Clearly, despite not being a leaf of the tree $\cal A$, the simple algorithm may blacken it. Hence, on including $p_i$ the effective $\mathit{count}$ of the parent of $p_i$ differs from that of its favoured descendant. Thus, each such $p_i\in P$ can be accounted for in the algorithm simply by increasing the count of its parent in $\cal A$ and treating it as \textit{favoured}.

\section{Experimental Evaluation}
\label{sec:exp}

We now evaluate the most promising algorithms for computing HOG in practice. Since the performance of these algorithms is strongly dependent on the underlying dataset, we evaluate these algorithms on both real and randomly generated datasets. Further, given the underlying application of HOG in bioinformatics, our real datasets were sources from EST datasets which have previously been used for the evaluation of algorithms solving APSP~\cite{HajRachid2015,DBLP:journals/tcs/LimP17}. Also, given the dependence of the complexity of various algorithms on $n$ and $k$, our generated random datasets help evaluate this dependence by varying over these parameters.

\subsection{Algorithms}
The following algorithms were found promising (see \Cref{sec:prevWork}) for evaluation in practice. 
\begin{enumerate}
    \item \texttt{CazauxR:} The simplest algorithm by Cazaux and Rivals\cite{DBLP:journals/ipl/CazauxR20} requiring $O(n+k^2)$ time. 
    \item \texttt{ParkCPR:} The improved algorithm by  Park et al. \cite{DBLP:conf/spire/ParkCPR20} requiring $O(n\log k)$ time. 
    \item \texttt{Khan:} The optimal algorithm by Khan~\cite{DBLP:conf/cpm/000421} requiring $O(n)$ time. 
    \item \texttt{ParkPCPR:} The optimal algorithm by Park et al.~\cite{DBLP:conf/cpm/ParkPCPR21} requiring $O(n)$ time. 
    \item \texttt{New:} Our proposed algorithm requiring $O(n)$ time.
\end{enumerate}
We followed the approach described in previous results~\cite{DBLP:conf/cpm/000421}, that it is better to first build an $\cal E$ from $\cal A$, and then use it to build $\cal H$ instead of directly building $\cal H$ using $\cal A$. Also, for a clearer comparison of relative performance, we use $\cal E$ as an input in our algorithms to avoid additional time used by all the algorithms. Recall that \texttt{ParkPCPR} cannot directly operate on $\cal E$ as an input, it additionally uses $P$ (or $\cal A$) to compute additional data structures. Finally, to store child pointers in $\cal E$ and $\cal A$ we use vectors and arrays respectively, instead of lists to improve performance (see details in \Cref{apn:impChoices}). 

\subsection{Performance measures and Environment}
Since all the algorithms compute HOG, which is unique for a given set of strings $P$, the only relevant parameters evaluated are the \textit{time} and \textit{memory}.  All the algorithms were implemented in C++ with g++ compiler (v9.4.0) and use $-O3$ optimization flag.  The performances were evaluated on 
an AMD Ryzen 9 7950X3D 16-Core processor having 32 cores and 128 GB RAM running Linux Ubuntu 20.04.6 LTS.

\subsection{Datasets}
The algorithms' performance primarily depends on the parameters $n$ and $k$. However, on further exploration, the relative size of underlying ACT, EHOG, and HOG generated were also seemingly significant. The algorithms were evaluated on the following datasets:
\begin{enumerate}
    \item \textbf{Real data.} We used the complete EST datasets\footnote{\url{http://www.citrusgenomedb.org/} and \url{http://www.uni-ulm.de/in/theo/research/seqana}} (over alphabet \{\textit{A, G, C, T}\}) of {Citrus clementina}, {Citrus sinensis}, {Citrus trifoliata2} and
{C. elegans} removing some low-quality reads (containing other characters), which were previously used for the evaluation of APSP algorithms~\cite{HajRachid2015,DBLP:journals/tcs/LimP17}. We additionally considered larger datasets on the genome assemblies of bacterial organisms\footnote{\url{http://ccb.jhu.edu/gage\_b/datasets/index.html}}, which were previously used for evaluation of genome assemblers~\cite{DBLP:journals/bioinformatics/MagocPCLSPTS13}.
See \Cref{tab:realDatasets} for a brief description of the datasets, which also shows the sizes of $\cal A, \cal E$, and $\cal H$, which also affect the performance of the algorithms (described later). 
\item \textbf{Random data.} Given the dependence of the algorithms on the parameters $n$ and $k$, we consider random datasets generated using two different approaches, namely,  (a) for a fixed value of $n=10^7$ and increasing values of $k$, and (b) for a fixed value of $k=10^5$, and increasing values of $n$. This allows us to evaluate the performance of the algorithms for each parameter, keeping the other as \textit{constant}. 

\item \textbf{Random Simulated Reads.} We consider a random string over $\{A,C,T,G\}$ of size $10^6$, and a set of its random substrings (reads) having $c\times$ coverage, having total length of reads $n=c\times 10^6$. For any coverage $c$, $n$ is constant, and $k$ is inversely proportional to the length of reads, which is varied for evaluation.  
\end{enumerate}

\noindent
\textbf{Remark:}
    The expected performance on random datasets is an average over $20$ test instances.


\begin{table}
\addtolength{\tabcolsep}{2pt}
\centering
\begin{tabular}{||c c c c c c||} 
 \hline
  Dataset & $k$ & $n$ & $|\cal A|$ & $|\cal E|$ & $|\cal H|$  \\ [0.5ex] 
 \hline\hline
trifoliata & 49.2K & 36.2M & 33.8M & 86.764K & 86.709K\\
\hline
clementina & 104.6K & 91.2M & 62.7M & 201.189K & 201.023K\\
\hline
sinensis & 151.9K & 107.3M & 101.0M & 311.208K & 311.014K\\
\hline
elegans & 206.5K & 108.3M & 81.0M & 867.630K & 867.618K\\[0.5ex] 
 \hline\hline
R\_sphaeroides\_M & 1.5M & 208.1M & 173.4M & 19.716M & 19.715M\\
\hline
V\_cholerae\_M & 1.6M & 293.2M & 235.8M & 30.610M & 30.609M\\
\hline
M\_abscessus\_M & 2.0M & 332.5M & 250.8M & 25.083M & 25.082M\\
\hline
B\_cereus\_M & 2.1M & 477.9M & 417.6M & 49.119M & 49.117M\\
\hline
V\_cholerae\_H & 3.9M & 353.7M & 208.1M & 58.798M & 58.798M\\
\hline
M\_abscessus\_H & 5.5M & 492.4M & 257.5M & 62.039M & 62.038M\\
\hline
S\_aureus\_H & 7.6M & 741.9M & 363.1M & 220.500M & 220.500M\\
\hline
R\_sphaeroides\_H & 7.7M & 583.5M & 343.9M & 132.872M & 132.870M\\
\hline
X\_axonopodis\_H & 11.7M & 1.1B & 614.6M & 281.072M & 281.071M\\
\hline
B\_cereus\_H & 12.0M & 1.2B & 599.8M & 315.366M & 315.363M\\
\hline
B\_fragilis\_H & 12.3M & 1.2B & 626.3M & 363.588M & 363.588M\\
\hline
A\_hydrophila\_H & 13.1M & 1.2B & 642.9M & 376.478M & 376.477M\\
\hline
\end{tabular}
\vspace{1em}
\caption{Real datasets from EST and genome assemblies.} 
\label{tab:realDatasets} 
\end{table}

\subsection{Evaluation on real datasets}

We evaluate the algorithms on the real datasets described in \Cref{tab:realDatasets}. \texttt{CazauxR} timed out ($\approx 12$hrs) on larger bacterial genome datasets. Hence, its evaluation is limited to EST datasets (see \Cref{tab:ActRelDumpSmall} in \Cref{apn:impChoices} for details), while the rest are evaluated on all the datasets (see \Cref{tab:RelPerfRealDatasets}). In both the datasets, we first compute $\cal E$ from the strings, which is used as an input to compute $\cal H$ by all the algorithms. Note that the time and memory to compute $\cal E$ is comparable to the running time of all algorithms when the dataset is large but extremely large for smaller datasets. This can be explained by comparing the relative sizes of $\cal A$ and $\cal E$ for EST ($100-300\times$) and bacterial genomes ($1.5-10\times$) as for computing $\cal E$ the input is $\cal A$, while $\cal H$ is computed using $\cal E$. We also note that the size of $\cal H$ only slightly improves $\cal E$, which indicates the significance of computing non-linear algorithms using $\cal E$ (computable linearly) instead of directly from $\cal A$. The results are shown relative to the best time and memory (actual performances are added to \Cref{apn:actP}).
 
For the EST datasets clearly, the proposed algorithm \texttt{New} performs faster compared to the other algorithm requiring minimum space as well, except in some cases where \texttt{ParkCPR} requires slightly less space. In general, as the size of the dataset increases the relative performance of \texttt{CazauxR} worsens around $~100-3000\times$ (see \Cref{tab:ActRelDumpSmall} in \Cref{apn:impChoices} for comparisons with \texttt{CazauxR}). On the other hand, the memory requirements are relatively same despite the super linear space complexity. For both \texttt{ParkCPR} and \texttt{Khan} the relative performance improves with the size of the dataset. This can be explained by the argument that as the size of the data increases the impact of the complexities of data structure involved may decrease. Another factor worth consideration is the relative memory usage of \texttt{ParkCPR} and \texttt{Khan}, where \texttt{Khan} required higher memory with respect to \texttt{ParkCPR} and \texttt{CazauxR}, despite improving the running time. 
Finally, the only algorithm unable to benefit from sharply reduced size of $\cal E$ over $\cal A$ is \texttt{ParkPCPR} which requires $60-220\times$ more memory, and time comparable only to \texttt{CazauxR} which is initially $2\times$ better but becomes upto $10\times$ worse as graph size increases. Overall, \texttt{ParkCPR} improves \texttt{CazauxR} by $20-180\times$, while \texttt{Khan} improves \texttt{ParkCPR} by $6-8\times$ using $3-5\times$ more memory, and \texttt{New} improves \texttt{Khan} by $1.5-1.7\times$ in time and by $3.5-4\times$ in memory. And \texttt{ParkPCPR} has mixed results only comparable to \texttt{CazauxR} improving it by $0.5-10\times$ in time at the expense of $60-200\times$ more memory.

\begin{table}[htbp]
\addtolength{\tabcolsep}{-2pt}
\centering
\begin{tabular}{|l|ll|ll|ll|ll|ll|}
\hline
\multirow{2}{*}{Dataset} & \multicolumn{2}{l|}{Computing $\cal E$}          & \multicolumn{2}{l|}{\texttt{New}}             & \multicolumn{2}{l|}{\texttt{Khan}}            & \multicolumn{2}{l|}{\texttt{ParkCPR}}         & \multicolumn{2}{l|}{\texttt{ParkPCPR}}         \\ \cline{2-11} 
                         & \multicolumn{1}{l|}{Mem}     & Time & \multicolumn{1}{l|}{Mem}   & Time & \multicolumn{1}{l|}{Mem}   & Time & \multicolumn{1}{l|}{Mem}   & Time & \multicolumn{1}{l|}{Mem}     & Time \\ \hline
clementina               & \multicolumn{1}{l|}{122.71x} & 61.93x  & \multicolumn{1}{l|}{1.00x} & 1.00x   & \multicolumn{1}{l|}{5.12x} & 1.69x   & \multicolumn{1}{l|}{1.06x} & 7.16x   & \multicolumn{1}{l|}{173.92x} & 202.23x \\ \hline
sinensis                 & \multicolumn{1}{l|}{132.17x} & 73.75x  & \multicolumn{1}{l|}{1.00x} & 1.00x   & \multicolumn{1}{l|}{5.05x} & 1.86x   & \multicolumn{1}{l|}{1.01x} & 7.77x   & \multicolumn{1}{l|}{176.54x} & 214.45x \\ \hline
trifoliata               & \multicolumn{1}{l|}{129.03x} & 86.28x  & \multicolumn{1}{l|}{1.00x} & 1.00x   & \multicolumn{1}{l|}{4.75x} & 1.86x   & \multicolumn{1}{l|}{1.07x} & 7.86x   & \multicolumn{1}{l|}{223.44x} & 313.62x \\ \hline
elegans                  & \multicolumn{1}{l|}{41.79x}  & 18.26x  & \multicolumn{1}{l|}{1.05x} & 1.00x   & \multicolumn{1}{l|}{3.35x} & 1.43x   & \multicolumn{1}{l|}{1.00x} & 5.47x   & \multicolumn{1}{l|}{60.60x}  & 58.39x  \\ \hline
A\_hydrophila\_H     & \multicolumn{1}{l|}{1.82x}   & 1.33x   & \multicolumn{1}{l|}{1.14x} & 1.00x   & \multicolumn{1}{l|}{1.82x} & 1.48x   & \multicolumn{1}{l|}{1.00x} & 3.98x   & \multicolumn{1}{l|}{2.68x}   & 1.84x   \\ \hline
B\_cereus\_H         & \multicolumn{1}{l|}{1.78x}   & 1.36x   & \multicolumn{1}{l|}{1.00x} & 1.00x   & \multicolumn{1}{l|}{1.56x} & 1.47x   & \multicolumn{1}{l|}{1.00x} & 3.92x   & \multicolumn{1}{l|}{2.41x}   & 1.97x   \\ \hline
B\_cereus\_M         & \multicolumn{1}{l|}{4.92x}   & 2.32x   & \multicolumn{1}{l|}{1.13x} & 1.00x   & \multicolumn{1}{l|}{1.94x} & 1.34x   & \multicolumn{1}{l|}{1.00x} & 3.51x   & \multicolumn{1}{l|}{6.81x}   & 4.90x   \\ \hline
B\_fragilis\_H       & \multicolumn{1}{l|}{1.82x}   & 1.00x   & \multicolumn{1}{l|}{1.11x} & 1.08x   & \multicolumn{1}{l|}{1.77x} & 1.34x   & \multicolumn{1}{l|}{1.00x} & 3.66x   & \multicolumn{1}{l|}{2.63x}   & 1.87x   \\ \hline
M\_abscessus\_H      & \multicolumn{1}{l|}{2.87x}   & 2.02x   & \multicolumn{1}{l|}{1.12x} & 1.00x   & \multicolumn{1}{l|}{2.23x} & 1.40x   & \multicolumn{1}{l|}{1.00x} & 4.16x   & \multicolumn{1}{l|}{4.51x}   & 3.10x   \\ \hline
M\_abscessus\_M      & \multicolumn{1}{l|}{5.59x}   & 2.55x   & \multicolumn{1}{l|}{1.12x} & 1.00x   & \multicolumn{1}{l|}{2.24x} & 1.38x   & \multicolumn{1}{l|}{1.00x} & 3.64x   & \multicolumn{1}{l|}{8.01x}   & 5.56x   \\ \hline
R\_sphaeroides\_H    & \multicolumn{1}{l|}{2.13x}   & 1.46x   & \multicolumn{1}{l|}{1.13x} & 1.00x   & \multicolumn{1}{l|}{2.07x} & 1.36x   & \multicolumn{1}{l|}{1.00x} & 4.19x   & \multicolumn{1}{l|}{3.29x}   & 2.00x   \\ \hline
R\_sphaeroides\_M    & \multicolumn{1}{l|}{4.56x}   & 1.82x   & \multicolumn{1}{l|}{1.00x} & 1.00x   & \multicolumn{1}{l|}{1.93x} & 1.34x   & \multicolumn{1}{l|}{1.00x} & 3.58x   & \multicolumn{1}{l|}{6.12x}   & 5.10x   \\ \hline
S\_aureus\_H         & \multicolumn{1}{l|}{1.69x}   & 1.00x   & \multicolumn{1}{l|}{1.14x} & 1.14x   & \multicolumn{1}{l|}{1.80x} & 1.41x   & \multicolumn{1}{l|}{1.00x} & 3.94x   & \multicolumn{1}{l|}{2.65x}   & 1.97x   \\ \hline
V\_cholerae\_H       & \multicolumn{1}{l|}{2.58x}   & 1.73x   & \multicolumn{1}{l|}{1.13x} & 1.00x   & \multicolumn{1}{l|}{2.05x} & 1.36x   & \multicolumn{1}{l|}{1.00x} & 3.57x   & \multicolumn{1}{l|}{3.98x}   & 2.70x   \\ \hline
V\_cholerae\_M       & \multicolumn{1}{l|}{4.49x}   & 2.02x   & \multicolumn{1}{l|}{1.13x} & 1.00x   & \multicolumn{1}{l|}{2.00x} & 1.35x   & \multicolumn{1}{l|}{1.00x} & 3.22x   & \multicolumn{1}{l|}{6.43x}   & 4.37x   \\ \hline
X\_axonopodis\_H     & \multicolumn{1}{l|}{1.80x}   & 1.45x   & \multicolumn{1}{l|}{1.00x} & 1.00x   & \multicolumn{1}{l|}{1.42x} & 1.36x   & \multicolumn{1}{l|}{1.00x} & 3.89x   & \multicolumn{1}{l|}{2.29x}   & 2.12x   \\ \hline
\end{tabular}
\vspace{1em}
\caption{Relative performance of algorithms on real datasets (actual performance in \cref{tab:ActPerfRealDatasets})}
\label{tab:RelPerfRealDatasets}
\end{table}

For the bacterial genome datasets, again \texttt{New} performs faster than other algorithms requiring less memory as well. Since \texttt{CazauxR} requires $>12$hrs it is improved by \texttt{ParkCPR} by at least $100\times$. For both \texttt{ParkCPR} and \texttt{Khan}, the relative performance (time and memory) again improves with the size of the dataset, where \texttt{Khan} still requires more memory despite being closer to \texttt{ParkCPR}.
Again, \texttt{ParkPCPR} requires $2-8\times$ more memory being unable to exploit computation through $\cal E$, and requires mixed time with respect to \texttt{ParkCPR} ranging from $0.5-1.5\times$. The only reason why its performance is comparable with faster algorithms is that the difference between $\cal A$ and $\cal E$ is way smaller than EST datasets. 
Overall, \texttt{Khan} improves \texttt{ParkCPR} by $1.3-2\times$ at the expense of slightly more memory, and \texttt{New} improves \texttt{Khan} by $1.3-1.5\times$ also improving the memory. 

To summarize minimum time is always taken by \texttt{New}  while minimum memory is often required by \texttt{ParkCPR} where \texttt{New} may be equivalent or up to $15\%$ higher.

\begin{observation}
On real EST and bacterial genome datasets evaluated  we have:
\begin{enumerate}
    \item \texttt{ParkCPR} improves \texttt{CazauxR} by $20$-$180\times$, \texttt{Khan} improves \texttt{ParkCPR} by $1.3$-$8\times$ using $1.3$-$5\times$ more memory, and \texttt{New} improves \texttt{Khan} in time by $1.3$-$1.7\times$ and in memory by $1.4$-$5\times$.
    \item With an increase in size of the dataset, the relative performances of all algorithms improve. 
    \item Though theoretically optimal \texttt{ParkPCPR} in unable to exploit computation using $\cal E$ making its performance strongly dependent on ratio of $|\cal A|$ and $|\cal E|$.
\end{enumerate}
\end{observation}

\subsection{Evaluation on random datasets}
For the random datasets, we measured the average running time of the algorithms over several generated strings having the same parameters of $n$ and $k$. The results of the experiments showing the variation of the running time with respect to $k$ and $n$ are shown in \Cref{fig:randomExp}. 

\begin{figure}
    \centering
    \includegraphics[scale=0.43]{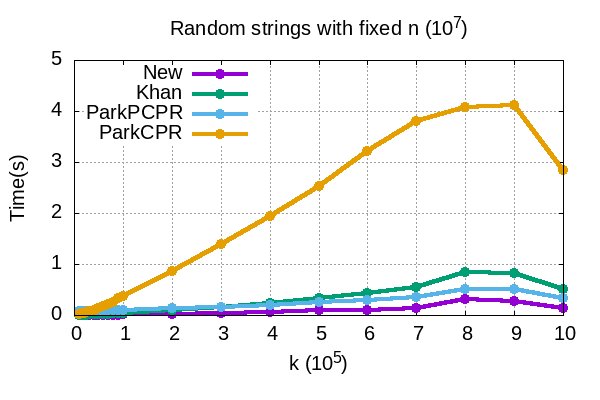}
    \includegraphics[scale=0.43]{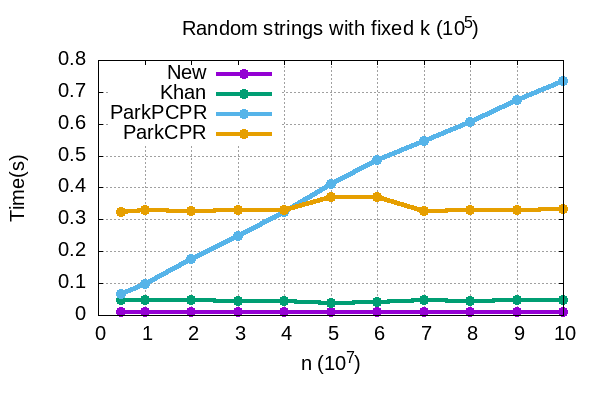}
\vspace{-1em}
    \caption{Variation of running time of the algorithms on the two types of random datasets}
    \label{fig:randomExp}
\end{figure}

The performance of the algorithms follow similar pattern as real datasets where \texttt{New}, \texttt{Khan} and \texttt{ParkCPR} are comparable while \texttt{CazauxR} is seemingly impractical. The performance of \texttt{ParkCPR} is improved by \texttt{Khan} by around $4\times$, while \texttt{Khan} is improved by \texttt{New} by $~3\times$. Again performance of \texttt{ParkPCPR} is mixed performing very well for fixed $k$ but very poorly for fixed $n$. This can again be attributed to the dependence of its performance on $|\cal A|$ instead of $|\cal E|$.

For a fixed value of $n=10^7$, as expected \texttt{CazauxR} increases quadratically with $k$. The other algorithms linearly increase up to $k\approx 8\times 10^5$, then decrease. Recall that \texttt{ParkCPR} had a logarithmic dependence on $k$, whereas \texttt{Khan} and \texttt{New} are independent of $k$. For such a large value of $k$ the lengths of individual string becomes closer to $10$, leading to smaller overlaps and hence smaller size of $\cal A, \cal E$ and $\cal H$. Since the algorithms use $\cal A,E$ as input, their size would greatly affect the performance of the algorithms. Also, due to this apparent decrease in $\cal A$, \texttt{ParkPCPR} performs $\approx 2\times$ better than \texttt{Khan} but $\approx 1.5\times$ slower than \texttt{New}.

For a fixed value of $k=10^5$, evaluating \texttt{CazauxR} is impractical due to a direct dependence on $k^2$, as was evident from the first experiment. For the remaining algorithms, we surprisingly see nearly \textit{constant} performance irrespective of $n$, which is very surprising because all the algorithms are linearly dependent on $n$. 
This is because, for completely random strings, the overlap is expected to be very small. Given that we have only four possible characters in the string $k$ uniformly random strings may have $\log k$ sized overlap which is independent of $n$. The small size overlaps result in the smaller size of $\cal E$. Since $k$ is the dominant factor in the size of $\cal E$, it possibly explains the near-constant performance of all the algorithms with the variation of $n$. However, this is not true for $\cal A$ which stores all the internal nodes not just corresponding to overlaps, which explains the linear dependence of \texttt{ParkPCPR} on $n$.

We verify the inference drawn from the above experiments by evaluating the variation in the size of $\cal E$ and $\cal A$ for both experiments (see \Cref{fig:randomExpEHOG}). Since the variation of $\cal E$ closely matches the performance of \texttt{ParkCPR}, \texttt{Khan} and \texttt{New}, it verifies our inferences that the running times of these algorithms are strongly dependent on $|\cal E|$ as compared to $n$ or $k$. 

\begin{figure}[t]
    \centering
    \includegraphics[scale=0.28]{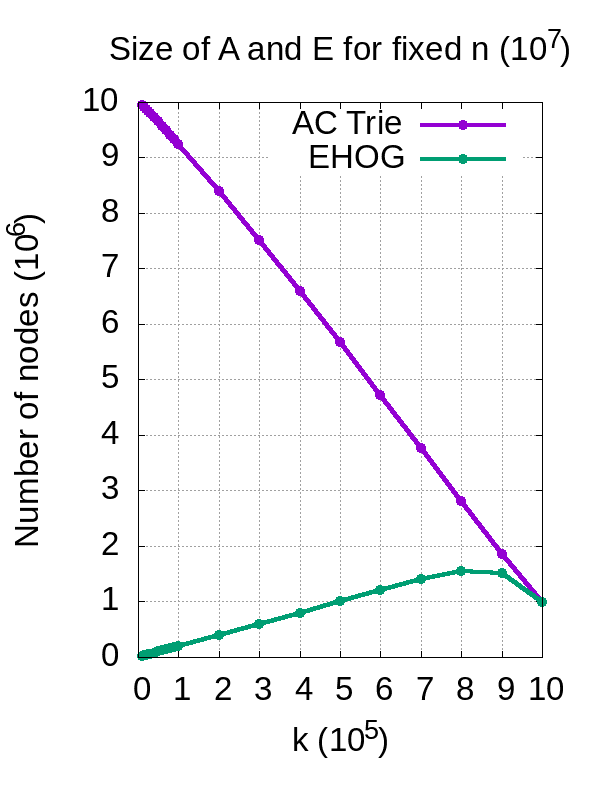}
     \includegraphics[scale=0.28]{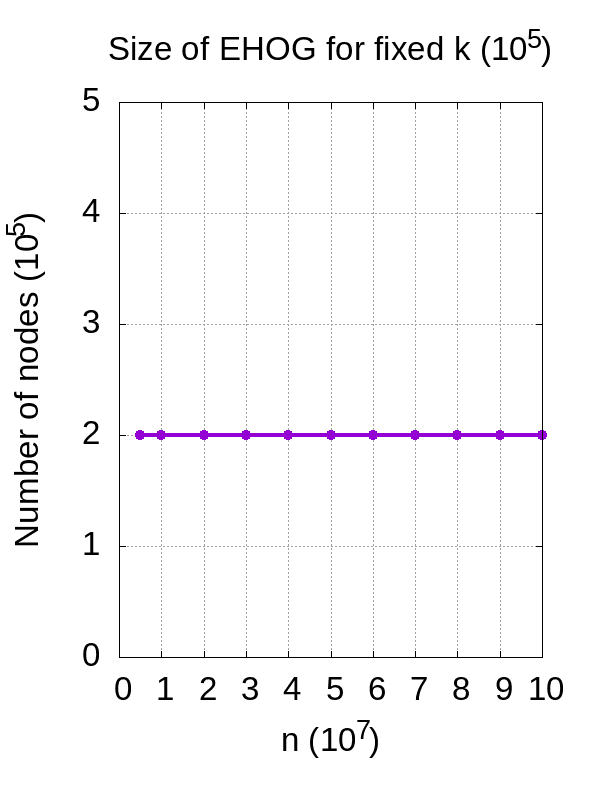}
    \includegraphics[scale=0.28]{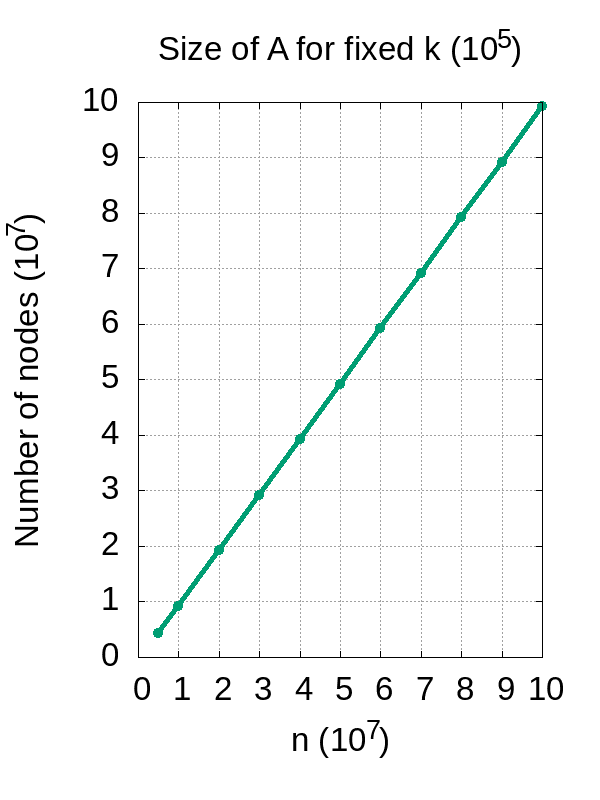}
    \vspace{-1em}
    \caption{Variation of the size of EHOG $\cal E$ and ACT $\cal A$ on the two types of random datasets. The left plot shows the variation of $k$, while mid and right for the variation of $n$.}
    \label{fig:randomExpEHOG}
\end{figure}

\begin{observation}
On random datasets evaluated under variation of $n$ and $k$ we have:
\begin{enumerate}
    \item \texttt{Khan} improves \texttt{ParkCPR} $4\times$, while \texttt{New} improves \texttt{Khan} by $~3\times$. 
    \item The performance of \texttt{ParkCPR}, \texttt{Khan} and \texttt{New} are strongly dependent on $|\cal E|$, while that of \texttt{ParkPCPR} is strongly dependent on $|\cal A|$, instead of $n,k$. 
\end{enumerate}
\end{observation}

\subsubsection{Evaluation on simulated random reads}

 We consider a random string of length $10^6$, and evaluate the performance of the algorithms (see~\Cref{fig:randomSim}) for varying length of snapshots (reads) for different coverage ($30\times, 50\times$ and $100\times$). As described earlier, for a given coverage $n$ is constant, and $k$ is inversely proportional to the length of reads,  which are varied for evaluation. The results are shown in \Cref{fig:randomSim}.

 Again, we see a similar pattern as other datasets. \texttt{New} is $~2\times$ better than \texttt{Khan} while \texttt{Khan} is $~2.5\times$ better than \texttt{ParkCPR}. Compared to these algorithms, \texttt{CazauxR} seems impractical. Some of the running times of \texttt{CazauxR} which were too high are not shown in \Cref{fig:randomSim} to improve the readability of the graph. Again, \texttt{ParkPCPR} is comparable to \texttt{New}, \texttt{Khan} and \texttt{ParkCPR}, but its performance depends on the sizes of $|\cal A|$ and  $|\cal E|$, which in turn depend on coverage and snapshot length.
 
 For a fixed snapshot length, as the coverage increases, $k$ and $n$ both increase. So, the size of $\cal A$ and $\cal E$ increases (see \Cref{fig:randomSimAho}). This causes the running time of all the algorithms to increase as well. So, the behavior is as expected.  
 For a fixed coverage, as the length of the snapshots increases, $k$ decreases, and $n$ stays constant. As, the total length of strings stays constant but the average length increases, the size of $\cal A$ increases very slowly as it weakly depends on the average string length. Also, the size of $\cal E$ sharply decreases because it depends quadratically on $k$. Thus, the running time of all algorithms, except \texttt{ParkPCPR}, decreases as they depend strongly on the size of $\cal E$. The running time of \texttt{ParkPCPR}, on the other hand, also decreases at first, but this decrease is slower compared to other algorithms, as \texttt{ParkPCPR} also depends on the size of $\cal A$, which is increasing. These inferences were verified by evaluating the variations in sizes of $\cal A$ and $\cal E$(see \Cref{fig:randomSimAho}). As they match with the variations in running times, our inferences are verified.

\begin{figure}[t]
    \centering
    \includegraphics[scale=0.45]{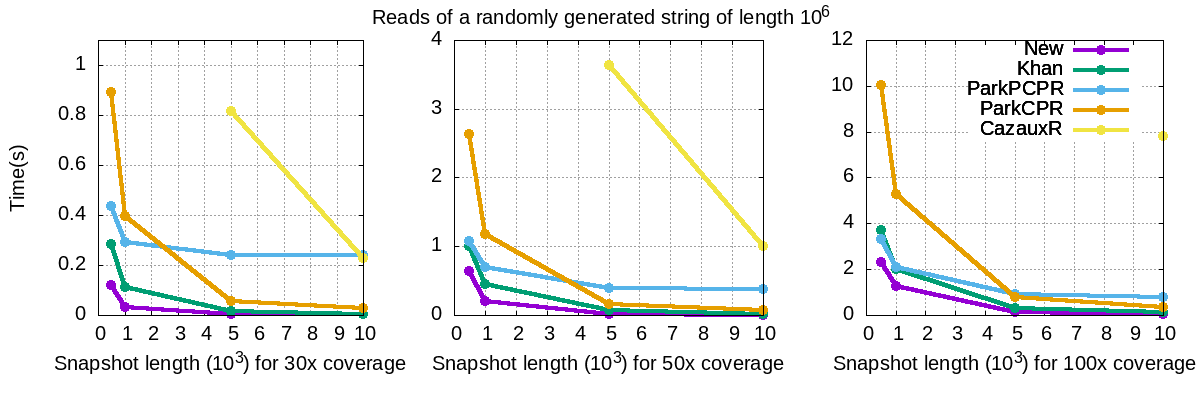}
    \vspace{-1em}
    \caption{Variation of running time for random simulated reads data.}        
    \label{fig:randomSim}
\end{figure}

\begin{figure}[h!]
    \centering
    \includegraphics[scale=0.45]{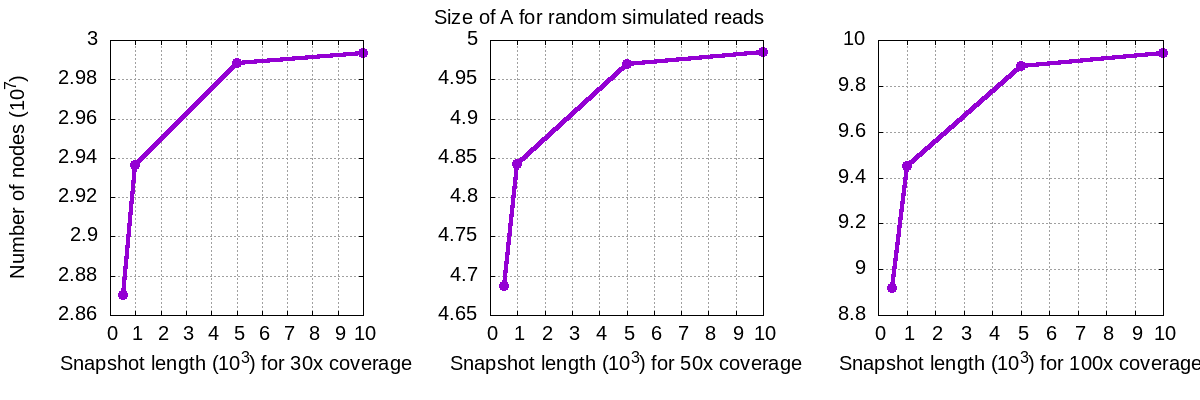}
    \includegraphics[scale=0.45]{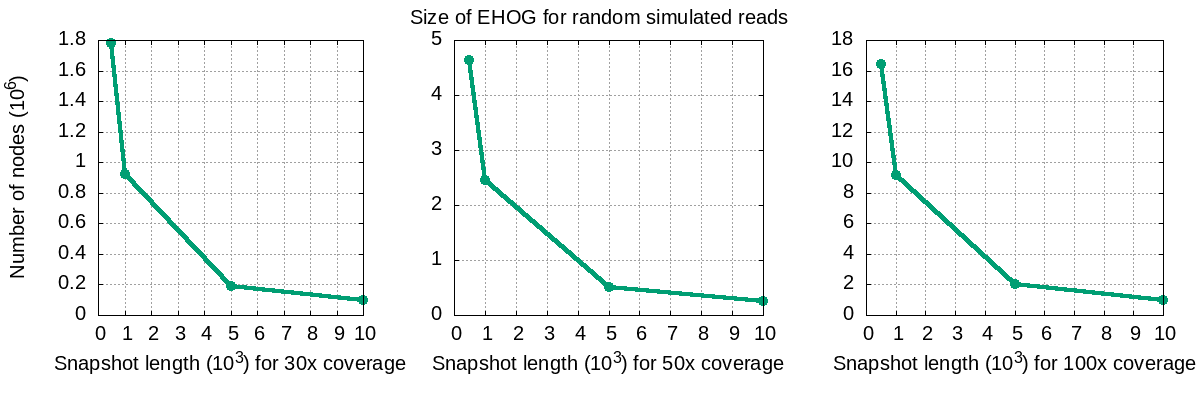}
    \caption{Variation of size of $\cal A$ and $\cal E$ for random simulated reads data.}
    \label{fig:randomSimAho}
\end{figure}


 \begin{observation}
On simulated random reads evaluated under variation of length and coverage, we have:
\begin{enumerate}
    \item \texttt{Khan} improves \texttt{ParkCPR} $~2\times$, while \texttt{New} improves \texttt{Khan} by $~2.5\times$. 
    \item The performance of \texttt{ParkCPR}, \texttt{Khan} and \texttt{New} are strongly dependent on $|\cal E|$, while that of \texttt{ParkPCPR} is strongly dependent on $|\cal A|$ as well as $|\cal E|$. 
\end{enumerate}
\end{observation}

 \newpage

\section{Applications to String Problems}
\label{sec:apps}

We consider variants of the All Pairs Suffix Prefix problem (APSP) which were recently studied in~\cite{DBLP:conf/cpm/LoukidesPTZ23}, giving state-of-the-art theoretical algorithms. 
However, they require complex black-box algorithms/data structures for their implementation, making them impractical for use in practice. We thus formulate simple solutions for the studied problems using HOG $\cal H$. 

Given a dictionary $P$ of $k$ strings $\{p_1, p_2, \cdots, p_k\}$ of total length $n$. We assume the strings in $P$ are lexicographically sorted; else, we sort them in $O(||P||)$ using ACT $\cal A$. We compute $\cal H$ and build the following data structure defined on each node $v\in \cal H$
\BlankLine
$SubTreeMinMax[v]:$ Minimum and Maximum index of strings in $P$ with prefix $v$
\BlankLine
Since the strings in $P$ are lexicographically sorted, each internal node in $x\in \cal H$ has a range of indices $[min,max]$ such that $\{p_{min},\cdots,p_{max}\}$ have $x$ as its prefix. Note that these strings are the leaves of the subtree rooted at $x$ in $\cal A$ or $\cal H$. Also, note that the root ($\epsilon$ or empty string) has this range $[1,n]$ having all the leaves. We can preprocess $\cal H$ to compute and store these ranges for all the internal nodes in linear time using linear space.

\subsection{Algorithms}
The studied problems can be categorized as



\begin{enumerate}
    \item \textbf{Paired String query:} 
    $OneToOne(i, j)$: report the string $ov(p_i,p_j)$
    \BlankLine
To compute the query, we start with $p_i$ and iterate over each node $v$ in its suffix path, evaluating whether $j\in [min,max]$ of $SubTreeMinMax[v]$. We simply report the first node for which the above property holds true. Since the range for root is $[1,k]$, the algorithm will always terminate with a solution. Also, we process each suffix $v$ of $p_i$ in decreasing order of length, the first $v$ that is also a prefix of $p_j$ is $ov(p_i,p_j)$. 
\BlankLine
Now, the number of nodes in the suffix path of a string is bounded by 
the length of the string $p_i$. Using precomputed data structure $SubTreeMinMax[v]$, each evaluation for $v$ takes $O(1)$ time, resulting in an overall complexity of $O(|p_i|)$. Also, note that the internal nodes in $\cal H$ are only the maximum overlaps between $k$ strings, limiting the internal nodes (and hence suffix path) to $O(k^2)$. This results in a complexity of $O(\min\{k^2,|p_i|\})$. Also, on average the length  $|p_i|$ is $O(n/k)$, so average cost of the query over all strings is $O(\min\{k^2,\frac{n}{k}\})$.

\BlankLine
\textbf{Note:} The above datastructure can also be computed for ACT $\cal A$ (or EHOG $\cal E$) and the corresponding algorithm will give the complexity $O(|p_i|)$. However, in practice, the difference in the number of nodes of $\cal A, E, H$ significantly impacts the performance even more than the theoretical limit of $O(k^2)$ for $\cal H$.
\BlankLine


    \item \textbf{Multi String Queries:} \\
    $OneToAll(i)$: report $ov(p_i,p_j)$ for every $p_j \in P$.\\
    $Report(i, l)$: report all $j\in[1,k]$ where $ov(p_i,p_j) \geq l$ for an integer $l \geq 0$.\\
    $Count(i, l)$: report the count of all $j\in [1,k]$  
    where $ov(p_i,p_j) \geq l$ for an integer $l \geq 0$.\\
    $Top(i, c)$: report all $j \in [1, k]$ corresponding to any $c$ highest values of $|ov(p_i,p_j)|$.  
 \BlankLine   
We first describe $OnetoAll(i)$ query where we need to find $ov(p_i,p_j)$, for every $p_j \in P$ for given $p_i$. We traverse each node $v$ in the suffix path of $p_i$, marking $v$ as the answer $ans[j]$ for each string $p_j$  in the subtree of $v$ (hence having prefix $v$), if not marked previously. Clearly, since we process suffix path in decreasing order of length of $v$, we store the maximum overlaps in $ans[\cdot]$ with $p_i$ (recall the simple process of blackening leaves described in \Cref{sec:simpPro}). 
\BlankLine
To perform this task efficiently, we store the next potentially unmarked leaf in $next[j]$ if $j$ is marked; else, we have $next[j]=j$. Hence, while processing $v$ on the suffix path, in case $next[j]=j$ for some leaf, we store $ans[j]=v$ and move to $next[j]$. Further, note that while processing $v$, none of the leaves in its subtree will be left unmarked and hence can be skipped while processing further nodes on the suffix path after $v$. This is achieved by updating $next[min]=max+1$ where $[min,max]$ is $SubTreeMinMax[v]$. This is sufficient to skip the above range in the future because any node $x$ on the suffix path trying to mark a leaf in the above range would surely be an ancestor of $v$, hence will first attempt to mark $min$ in this range thereby skipping processing up to $max+1$.
\BlankLine
Again, the number of nodes in the suffix path is bounded by $O(|p_i|)$, 
and we use total $O(k+|p_i|)$ time to mark and skip using $next[\cdot]$. For each $v$ in the suffix path, it visits $next[i]$ of either the minimum index, or of a leaf that will never be processed again (skipped in the future). Thus, $next[\cdot]$ is processed once for each leaf and once for each node in the suffix path, requiring $O(k+|p_i|)$ time. Again, the length of the suffix path is also bounded by the number of internal nodes in $\cal H$, i.e., $O(k^2)$ resulting in the complexity of $O(k+\min\{k^2,|p_i|\})$. And the average cost of the query across all strings is $O(k+\min\{k^2,\frac{n}{k}\})$.

\BlankLine
The queries $Report(i, l), Count(i,l)$, and $Top(i,c)$ terminate the above process early when the length of string $v$ on suffix path becomes smaller than $l$, or number of reported strings exceed $c$. Thus, the above queries have the same time complexity in the worst case.

\end{enumerate}

\begin{table}
\centering
\addtolength{\tabcolsep}{2pt}
\vspace{-1em}
\begin{tabular}{||c c c c||} 
 \hline
Problem & Classical~\cite{DBLP:journals/siamcomp/KnuthMP77} & State of the art~\cite{DBLP:conf/cpm/LoukidesPTZ23} & Using HOG  \\ [0.5ex] 
 \hline\hline
$OneToOne(i, j)$ & $O(|p_i| + |p_j|)$ & $O({\log{\log{k}}})$ & $O(\min\{k^2,|p_i|\})$ \\ 
$OneToAll(i)$ & $O(k|p_i|+n)$ & $O(k)$  & $O(k+\min\{k^2,|p_i|\})$ \\ 
$Report(i, l)$ & $O(k|p_i|+n)$ & $O(\log{n}/\log{\log{n}} + output)$ & $O(k+\min\{k^2,|p_i|\})$ \\ 
$Count(i, l)$ &  $O(k|p_i|+n)$ & $O(\log{n}/\log{\log{n}})$  & $O(k+\min\{k^2,|p_i|\})$ \\ 
$Top(i, c)$ & $O(k|p_i|+n)$ & $O(\log^2{n}/\log{\log{n}} + output)$& $O(k+\min\{k^2,|p_i|\})$ \\ 
 \hline
\end{tabular}
\vspace{1em}
\caption{Comparison of the classical and state-of-the-art algorithms with the HOG algorithms.}
\vspace{-3em}
 \label{tab:apspAlgo2}
\end{table}

\subsection{Classical Algorithm}
Note that each overlap query $ov(p_i,p_j)$ can be naively computed by comparing every suffix of $p_i$ with $p_j$, requiring $O(|p_i|\times |p_j|)$ time. However, using the classical KMP algorithm~\cite{DBLP:journals/siamcomp/KnuthMP77} with pattern $p_j$ and text $p_i$ it can be answered in $O(|p_i|+|p_j|)$. For multiple string queries, we repeat for all $p_j\in P$, requiring total $O(\sum_{j=1}^{k} |p_i|+|p_j|)= O(k|p_i|+n)$ time. Note that the average cost of the query over all possible strings $p_i$ is $O(k\times \frac{n}{k}+n)=O(n)$.

\subsection{Comparisons with the state-of-the-art}
Current algorithms with HOG are easier to implement than \cite{DBLP:conf/cpm/LoukidesPTZ23}, though clearly, it has a lot of scope for improvement. Also, we are able to answer multiple types of queries with only HOG, while in \cite{DBLP:conf/cpm/LoukidesPTZ23} they have used different data structures for different types of queries which makes them harder to implement and will practically take more time for construction. 
\Cref{tab:apspAlgo2} summarizes the complexities of the proposed algorithms for the studied problems and the pseduocodes of our algorithms are given in \Cref{apn:pseudocodes}.


\subsection{Experimental Evaluation}


Despite failing to match the state-of-the-art algorithms~\cite{DBLP:conf/cpm/LoukidesPTZ23} theoretically by a large margin, our algorithms performed well in practice (see \Cref{tab:realQueriesItalian}). To answer paired string queries our algorithm required $0.002-0.010~ms$ and for multi string queries $0.5-100~ms$ for a data set having around a billion characters. For comparison (see \Cref{tab:realQueriesKMP} in \Cref{apn:actP}), the classical KMP algorithm answers paired string queries $2-20\times$ faster ($0.0005-0.0009~ms$), where one-to-one queries are tailor-made for KMP algorithm. However, the remaining multi string queries $18-1300\times$ slower ($600-6500~ms$). 
Among the multi string queries one-to-all has the least improvement ($18-106\times$), and count has the most improvement ($199-1300\times$), as KMP is unable to exploit any structural information having to compute all overlaps despite not being asked.  Since the state-of-the-art algorithms~\cite{DBLP:conf/cpm/LoukidesPTZ23} require complex black box algorithms and data structures, its evaluation is beyond the scope of this paper.  

\begin{table}
\centering
\addtolength{\tabcolsep}{-2pt}
\begin{tabular}{||c c c c c c c c||} 
 \hline
  Dataset & $k$ & $n$ & one-to-one & one-to-all & top & count & report  \\ [0.5ex] 
 \hline\hline
A\_hydrophila\_H & 13081385 & 1235604318 & 0.010 & 95.375 & 34.238 & 8.434 & 10.783\\
\hline
B\_cereus\_H & 12018488 & 1161420605 & 0.009 & 71.926 & 30.076 & 4.433 & 8.397\\
\hline
B\_cereus\_M & 2068138 & 477906536 & 0.003 & 8.654 & 4.070 & 0.698 & 0.922\\
\hline
B\_fragilis\_H & 12317015 & 1198165883 & 0.008 & 74.648 & 30.248 & 4.583 & 6.302\\
\hline
M\_abscessus\_H & 5507950 & 492353747 & 0.002 & 28.556 & 11.384 & 1.627 & 2.376\\
\hline
M\_abscessus\_M & 1986962 & 332468005 & 0.002 & 7.945 & 4.026 & 0.656 & 0.883\\
\hline
R\_sphaeroides\_H & 7727979 & 583472776 & 0.003 & 51.254 & 19.836 & 2.811 & 4.310\\
\hline
R\_sphaeroides\_M & 1509743 & 208072134 & 0.002 & 5.954 & 2.925 & 0.516 & 0.722\\
\hline
S\_aureus\_H & 7601994 & 741898249 & 0.005 & 32.004 & 13.138 & 1.832 & 2.729\\
\hline
V\_cholerae\_H & 3854210 & 353704636 & 0.002 & 15.161 & 6.823 & 1.085 & 1.451\\
\hline
V\_cholerae\_M & 1577183 & 293246651 & 0.002 & 5.826 & 2.967 & 0.510 & 0.670\\
\hline
X\_axonopodis\_H & 11669945 & 1096782704 & 0.009 & 76.691 & 30.474 & 5.002 & 7.976\\
\hline
\end{tabular}
\caption{Running time (in ms) of queries for genome assemblies datasets.} \label{tab:realQueriesItalian} 
\end{table}

\section{Conclusion}
\label{sec:conclusion}
We proposed a new algorithm for computing HOG and performed an empirical evaluation of all the practical algorithms. While the previous optimal algorithm~\cite{DBLP:conf/cpm/000421} improves the other algorithms in time, it comes at the expense of increased memory usage. Our algorithm improves both time and memory with respect to all previous algorithms. Moreover, our algorithm is also arguably more intuitive and easier to implement using only elementary arrays. To highlight the significance of our results, we also considered some applications of the HOG on variants of the APSP problem and demonstrated its practicality by showing that it requires acceptable time despite the large size of the dataset. 
In the future, faster algorithms based on HOG for the studied variants of APSP might prove useful. Possibly, the structure of HOG can result in superior algorithms both theoretically and in practice.


\bibliography{paper}

\appendix

\section{Justification of Implementation Approach}
\label{apn:impChoices}
 







\begin{table}
\centering
\addtolength{\tabcolsep}{2pt}
\begin{tabular}{|l l l | l | l | l | l |}
\hline
\multicolumn{3}{|l|}{Dataset}                                                                                 & clementina          & sinensis            & trifoliata          & elegans             \\ \hline
\multicolumn{1}{|l|}{\multirow{4}{*}{\texttt{New}}}      & \multicolumn{1}{l|}{\multirow{2}{*}{Via $\cal A$}} & Mem     & 3.37G               & 5.39G               & 1.80G               & 4.35G               \\ \cline{3-7} 
\multicolumn{1}{|l|}{}                          & \multicolumn{1}{l|}{}                              & Time(s) & 3.11                & 5.37                & 1.75                & 4.42                \\ \cline{2-7} 
\multicolumn{1}{|l|}{}                          & \multicolumn{1}{l|}{\multirow{2}{*}{Via $\cal E$}} & Mem     & 2.38G               & 3.80G               & 1.27G               & 3.11G               \\ \cline{3-7} 
\multicolumn{1}{|l|}{}                          & \multicolumn{1}{l|}{}                              & Time(s) & 3.33                & 5.84                & 1.83                & 4.65                \\ \hline
\multicolumn{1}{|l|}{\multirow{4}{*}{\texttt{Khan}}}     & \multicolumn{1}{l|}{\multirow{2}{*}{Via $\cal A$}} & Mem     & 3.70G               & 5.91G               & 1.98G               & 4.84G               \\ \cline{3-7} 
\multicolumn{1}{|l|}{}                          & \multicolumn{1}{l|}{}                              & Time(s) & 3.76                & 6.35                & 1.98                & 5.07                \\ \cline{2-7} 
\multicolumn{1}{|l|}{}                          & \multicolumn{1}{l|}{\multirow{2}{*}{Via $\cal E$}} & Mem     & 2.38G               & 3.80G               & 1.27G               & 3.11G               \\ \cline{3-7} 
\multicolumn{1}{|l|}{}                          & \multicolumn{1}{l|}{}                              & Time(s) & 3.54                & 6.03                & 1.76                & 4.81                \\ \hline
\multicolumn{1}{|l|}{\multirow{4}{*}{\texttt{ParkCPR}}}  & \multicolumn{1}{l|}{\multirow{2}{*}{Via $\cal A$}} & Mem     & 2.62G               & 4.18G               & 1.40G               & 3.38G               \\ \cline{3-7} 
\multicolumn{1}{|l|}{}                          & \multicolumn{1}{l|}{}                              & Time(s) & 4.07                & 6.82                & 2.08                & 6.11                \\ \cline{2-7} 
\multicolumn{1}{|l|}{}                          & \multicolumn{1}{l|}{\multirow{2}{*}{Via $\cal E$}} & Mem     & 2.38G               & 3.80G               & 1.27G               & 3.11G               \\ \cline{3-7} 
\multicolumn{1}{|l|}{}                          & \multicolumn{1}{l|}{}                              & Time(s) & 3.90                & 6.60                & 2.01                & 5.80                \\ \hline
\multicolumn{1}{|l|}{\multirow{4}{*}{\texttt{ParkPCPR}}} & \multicolumn{1}{l|}{\multirow{2}{*}{Via $\cal A$}} & Mem     & 5.12G               & 7.88G               & 2.64G               & 6.54G               \\ \cline{3-7} 
\multicolumn{1}{|l|}{}                          & \multicolumn{1}{l|}{}                              & Time(s) & 3.83                & 6.14                & 1.93                & 4.96                \\ \cline{2-7} 
\multicolumn{1}{|l|}{}                          & \multicolumn{1}{l|}{\multirow{2}{*}{Via $\cal E$}} & Mem     & 3.38G               & 5.07G               & 1.70G               & 4.36G               \\ \cline{3-7} 
\multicolumn{1}{|l|}{}                          & \multicolumn{1}{l|}{}                              & Time(s) & 4.28                & 6.89                & 2.16                & 5.66                \\ \hline
\multicolumn{1}{|l|}{\multirow{4}{*}{\texttt{CazauxR}}}  & \multicolumn{1}{l|}{\multirow{2}{*}{Via $\cal A$}} & Mem     & -                   & -                   & -                   & -                   \\ \cline{3-7} 
\multicolumn{1}{|l|}{}                          & \multicolumn{1}{l|}{}                              & Time(s) & \textgreater{}30min & \textgreater{}30min & \textgreater{}30min & \textgreater{}30min \\ \cline{2-7} 
\multicolumn{1}{|l|}{}                          & \multicolumn{1}{l|}{\multirow{2}{*}{Via $\cal E$}} & Mem     & 2.21G               & 3.53G               & 1.18G               & 2.89                \\ \cline{3-7} 
\multicolumn{1}{|l|}{}                          & \multicolumn{1}{l|}{}                              & Time(s) & 19.37               & 41.73               & 5.05                & 132.26              \\ \hline
\end{tabular}
\vspace{1em}
\caption{Actual Performance of the algorithms on bacterial genome datasets using input $P$} 
\label{tab:actualPerformanceGA} 
\vspace{-2em}
\end{table}

The computation of $\cal H$ can be performed directly using the strings $P$ and hence $\cal A$ or from $\cal E$ which is claimed to be more scalable. \Cref{tab:actualPerformanceGA} shows a comparison between the two approaches where less memory is required when computing via $\cal E$ as compared to $\cal A$ sometimes at the expense of slightly higher running time. 

\begin{table}[htbp]
\vspace{-2em}
\centering
\scalebox{1}{\begin{tblr}{
  cell{1}{1} = {c=3}{},
  cell{2}{1} = {r=4}{},
  cell{2}{2} = {r=2}{},
  cell{4}{2} = {r=2}{},
  cell{6}{1} = {r=4}{},
  cell{6}{2} = {r=2}{},
  cell{8}{2} = {r=2}{},
  cell{10}{1} = {r=4}{},
  cell{10}{2} = {r=2}{},
  cell{12}{2} = {r=2}{},
  cell{14}{1} = {r=4}{},
  cell{14}{2} = {r=2}{},
  cell{16}{2} = {r=2}{},
  cell{18}{1} = {r=4}{},
  cell{18}{2} = {r=2}{},
  cell{20}{2} = {r=2}{},
  cell{22}{1} = {r=4}{},
  cell{22}{2} = {r=2}{},
  cell{24}{2} = {r=2}{},
  vlines,
  hline{1-2,6,10,14,18,22,26} = {-}{},
  hline{3,5,7,9,11,13,15,17,19,21,23,25} = {3-7}{},
  hline{4,8,12,16,20,24} = {2-7}{},
}
Dataset                   &          &         & clementina & sinensis & trifoliata & elegans \\ Computing
$\cal E$ & Actual   & Mem     & 2.38G      & 3.80G    & 1.27G      & 3.11G   \\
                          &          & Time(s) & 3.37       & 6.01     & 1.86       & 4.63    \\
                          & Relative & Mem     & 157.75x    & 132.17x  & 129.03x    & 41.79x  \\
                          &          & Time & 61.93x     & 73.75x   & 86.28x     & 18.26x  \\
\texttt{New}                       & Actual   & Mem     & 19.41M     & 28.73M   & 9.84M      & 78.45M  \\
                          &          & Time(s) & 0.05       & 0.08     & 0.02       & 0.25    \\
                          & Relative & Mem     & 1.29x      & 1.00x    & 1.00x      & 1.05x   \\
                          &          & Time & 1.00x      & 1.00x    & 1.00x      & 1.00x   \\
\texttt{Khan}                      & Actual   & Mem     & 99.44M     & 145.05M  & 46.77M     & 249.32M \\
                          &          & Time(s) & 0.09       & 0.15     & 0.04       & 0.36    \\
                          & Relative & Mem     & 6.59x      & 5.05x    & 4.75x      & 3.35x   \\
                          &          & Time & 1.69x      & 1.86x    & 1.86x      & 1.43x   \\
\texttt{ParkCPR}                   & Actual   & Mem     & 20.50M     & 29.10M   & 10.54M     & 74.42M  \\
                          &          & Time(s) & 0.39       & 0.63     & 0.17       & 1.39    \\
                          & Relative & Mem     & 1.36x      & 1.01x    & 1.07x      & 1.00x   \\
                          &          & Time & 7.16x      & 7.77x    & 7.86x      & 5.47x   \\
\texttt{ParkPCPR}                  & Actual   & Mem     & 3.38G      & 5.07G    & 2.20G      & 4.51G   \\
                          &          & Time(s) & 11.01      & 17.48    & 6.76       & 14.81   \\
                          & Relative & Mem     & 223.58x    & 176.54x  & 223.44x    & 60.60x  \\
                          &          & Time & 202.23x    & 214.45x  & 313.62x    & 58.39x  \\
\texttt{CazauxR}                   & Actual   & Mem     & 15.10M     & 32.30M   & 47.40M     & 134.53M \\
                          &          & Time(s) & 5.36       & 24.04    & 59.81      & 195.86  \\
                          & Relative & Mem     & 1.00x      & 1.12x    & 4.82x      & 1.81x   \\
                          &          & Time & 98.41x     & 294.86x  & 2.77Kx     & 772.40x 
\end{tblr}}
\vspace{1em}
\caption{Actual and Relative performance of the algorithm using $\cal E$ as input.}
\vspace{-4em}
\label{tab:ActRelDumpSmall}
\end{table}

We further note that a significant part of the time is required in computing $\cal A$ and $\cal E$ which is computed by all the algorithms. This makes the relative performances of the algorithms less pronounced. A clearer picture would be possible in case we use $\cal E$ directly as an input. The corresponding difference is highlighted in \Cref{tab:ActRelDumpSmall}.

Finally, we tried a linked list to store children of both $\cal A$ and $\cal E$ nodes. In this case, we are converting the $\cal A$ nodes directly into the $\cal E$ while changing the links only. We delete additional $\cal A$ nodes that are not included in the $\cal E$. Also, we used a pointer-based in which We store $\cal A$ children into an array of size of number of characters in the dataset and $\cal E$ children into a dynamic vector. Then, We make $\cal A$ using pointers of the $\cal A$ nodes. While constructing $\cal E$, We are making new $\cal E$ nodes and deleting the $\cal A$ nodes side by side while DFS the $\cal A$. Then, finally, we store $\cal A$ as a vector of $\cal A$ nodes. In this, we are not deleting the $\cal A$ nodes at all. As in \Cref{tab:difApproachComp}, We tested all these three approaches using A\_hydophila\_H Dataset and \texttt{New} algorithm, The vector of $\cal A$ nodes is the fastest and most memory-efficient approach in all three approaches. So, We have used this approach in the experimentation.

\begin{table}
\vspace{-1em}
\centering
\addtolength{\tabcolsep}{2pt}
\begin{tabular}{|l|l|l|l|l|}
\hline
                              &         & Array Based & Pointer Based & List Based \\ \hline
\multirow{2}{*}{Via $\cal A$} & Mem     & 35.37G      & 47.31G        & 109.14G    \\ \cline{2-5} 
                              & Time(s) & 140.32      & 244.36        & 360.86     \\ \hline
\multirow{2}{*}{Via $\cal E$} & Mem     & 56.90G      & 57.78G        & 104.34G    \\ \cline{2-5} 
                              & Time(s) & 218.91      & 388.87        & 528.76     \\ \hline
\end{tabular}
\vspace{1em}
\caption{Comparison between different implementations using A\_hydrophila\_H Dataset} 
\label{tab:difApproachComp} 
\end{table}

\section{Actual performance of algorithms on real datasets}
\label{apn:actP}

\begin{table}[h]
\centering
\addtolength{\tabcolsep}{-5pt}
\begin{tabular}{|l|lr|lr|lr|lr|lr|}
\hline
\multirow{2}{*}{Dataset} & \multicolumn{2}{l|}{Computing $\cal E$}                              & \multicolumn{2}{l|}{\texttt{New}}                                   & \multicolumn{2}{l|}{\texttt{Khan}}                                   & \multicolumn{2}{l|}{\texttt{ParkCPR}}                               & \multicolumn{2}{l|}{\texttt{ParkPCPR}}                              \\ \cline{2-11} 
                         & \multicolumn{1}{l|}{Mem}    & \multicolumn{1}{l|}{Time(s)} & \multicolumn{1}{l|}{Mem}    & \multicolumn{1}{l|}{Time(s)} & \multicolumn{1}{l|}{Mem}     & \multicolumn{1}{l|}{Time(s)} & \multicolumn{1}{l|}{Mem}    & \multicolumn{1}{l|}{Time(s)} & \multicolumn{1}{l|}{Mem}    & \multicolumn{1}{l|}{Time(s)} \\ \hline
clementina               & \multicolumn{1}{l|}{2.38G}  & 3.37                         & \multicolumn{1}{l|}{19.41M} & 0.05                         & \multicolumn{1}{l|}{99.44M}  & 0.09                         & \multicolumn{1}{l|}{20.50M} & 0.39                         & \multicolumn{1}{l|}{3.38G}  & 11.01                        \\ \hline
sinensis                 & \multicolumn{1}{l|}{3.80G}  & 6.01                         & \multicolumn{1}{l|}{28.73M} & 0.08                         & \multicolumn{1}{l|}{145.05M} & 0.15                         & \multicolumn{1}{l|}{29.10M} & 0.63                         & \multicolumn{1}{l|}{5.07G}  & 17.48                        \\ \hline
trifoliata               & \multicolumn{1}{l|}{1.27G}  & 1.86                         & \multicolumn{1}{l|}{9.84M}  & 0.02                         & \multicolumn{1}{l|}{46.77M}  & 0.04                         & \multicolumn{1}{l|}{10.54M} & 0.17                         & \multicolumn{1}{l|}{2.20G}  & 6.76                         \\ \hline
elegans                  & \multicolumn{1}{l|}{3.11G}  & 4.63                         & \multicolumn{1}{l|}{78.45M} & 0.25                         & \multicolumn{1}{l|}{249.32M} & 0.36                         & \multicolumn{1}{l|}{74.42M} & 1.39                         & \multicolumn{1}{l|}{4.51G}  & 14.81                        \\ \hline
A\_hydrophila\_H     & \multicolumn{1}{l|}{54.99G} & 226.8                        & \multicolumn{1}{l|}{34.47G} & 170.85                       & \multicolumn{1}{l|}{54.96G}  & 253.51                       & \multicolumn{1}{l|}{30.25G} & 680.75                       & \multicolumn{1}{l|}{81.11G} & 314.12                       \\ \hline
B\_cereus\_H         & \multicolumn{1}{l|}{53.29G} & 185.76                       & \multicolumn{1}{l|}{29.86G} & 136.31                       & \multicolumn{1}{l|}{46.55G}  & 200.12                       & \multicolumn{1}{l|}{29.86G} & 534.82                       & \multicolumn{1}{l|}{71.99G} & 268.18                       \\ \hline
B\_cereus\_M         & \multicolumn{1}{l|}{19.41G} & 46.2                         & \multicolumn{1}{l|}{4.47G}  & 19.9                         & \multicolumn{1}{l|}{7.64G}   & 26.75                        & \multicolumn{1}{l|}{3.95G}  & 69.81                        & \multicolumn{1}{l|}{26.87G} & 97.5                         \\ \hline
B\_fragilis\_H       & \multicolumn{1}{l|}{54.32G} & 158.43                       & \multicolumn{1}{l|}{33.29G} & 171.11                       & \multicolumn{1}{l|}{52.83G}  & 212.9                        & \multicolumn{1}{l|}{29.89G} & 579.89                       & \multicolumn{1}{l|}{78.63G} & 295.73                       \\ \hline
M\_abscessus\_H      & \multicolumn{1}{l|}{14.44G} & 51.43                        & \multicolumn{1}{l|}{5.62G}  & 25.44                        & \multicolumn{1}{l|}{11.22G}  & 35.56                        & \multicolumn{1}{l|}{5.03G}  & 105.8                        & \multicolumn{1}{l|}{22.64G} & 78.75                        \\ \hline
M\_abscessus\_M      & \multicolumn{1}{l|}{11.35G} & 26.57                        & \multicolumn{1}{l|}{2.27G}  & 10.43                        & \multicolumn{1}{l|}{4.54G}   & 14.43                        & \multicolumn{1}{l|}{2.03G}  & 37.96                        & \multicolumn{1}{l|}{16.27G} & 58.01                        \\ \hline
R\_sphaeroides\_H    & \multicolumn{1}{l|}{22.83G} & 99.02                        & \multicolumn{1}{l|}{12.10G} & 67.89                        & \multicolumn{1}{l|}{22.17G}  & 92.39                        & \multicolumn{1}{l|}{10.73G} & 284.49                       & \multicolumn{1}{l|}{35.36G} & 135.88                       \\ \hline
R\_sphaeroides\_M    & \multicolumn{1}{l|}{8.40G}  & 15.61                        & \multicolumn{1}{l|}{1.84G}  & 8.6                          & \multicolumn{1}{l|}{3.55G}   & 11.49                        & \multicolumn{1}{l|}{1.84G}  & 30.77                        & \multicolumn{1}{l|}{11.28G} & 43.85                        \\ \hline
S\_aureus\_H         & \multicolumn{1}{l|}{29.97G} & 86.69                        & \multicolumn{1}{l|}{20.19G} & 98.87                        & \multicolumn{1}{l|}{31.92G}  & 121.93                       & \multicolumn{1}{l|}{17.72G} & 341.56                       & \multicolumn{1}{l|}{46.96G} & 170.71                       \\ \hline
V\_cholerae\_H       & \multicolumn{1}{l|}{12.23G} & 41.63                        & \multicolumn{1}{l|}{5.35G}  & 24.12                        & \multicolumn{1}{l|}{9.73G}   & 32.72                        & \multicolumn{1}{l|}{4.74G}  & 86.1                         & \multicolumn{1}{l|}{18.87G} & 65.1                         \\ \hline
V\_cholerae\_M       & \multicolumn{1}{l|}{11.07G} & 24.69                        & \multicolumn{1}{l|}{2.79G}  & 12.22                        & \multicolumn{1}{l|}{4.93G}   & 16.44                        & \multicolumn{1}{l|}{2.47G}  & 39.32                        & \multicolumn{1}{l|}{15.87G} & 53.36                        \\ \hline
X\_axonopodis\_H     & \multicolumn{1}{l|}{53.72G} & 180.35                       & \multicolumn{1}{l|}{29.82G} & 124.58                       & \multicolumn{1}{l|}{42.38G}  & 169.8                        & \multicolumn{1}{l|}{29.82G} & 484.43                       & \multicolumn{1}{l|}{68.32G} & 263.64                       \\ \hline
\end{tabular}
\label{tab:ActPerfRealDatasets}
\vspace{1em}
\caption{Actual Performance of Algorithms using Real Datasets}
\vspace{-2em}
\end{table}

\begin{table}
\centering
\addtolength{\tabcolsep}{-4pt}
\begin{tabular}{||c c c c c c c c c||} 
 \hline
  Dataset & $k$ & $n$  & Algorithm & one-to-one & one-to-all & top & count & report  \\ [0.5ex] 
 \hline\hline
A\_hydrophila\_H& 13081385 & 1235604318 & HOG & 0.010 & 95.375 & 34.238 & 8.434 & 10.783 \\ 
 & & & KMP & 0.00045 & 1703.225 & 6515.548 & 1684.589 & 1689.835 \\ 
 & & & Rel & 0.05$\times$ & 17.86$\times$ & 190.30$\times$ & 199.74$\times$ & 156.71$\times$ \\ \hline 
B\_cereus\_H& 12018488 & 1161420605 & HOG & 0.009 & 71.926 & 30.076 & 4.433 & 8.397 \\ 
 & & & KMP & 0.00045 & 1439.163 & 5421.761 & 1376.226 & 1376.502 \\ 
 & & & Rel & 0.05$\times$ & 20.01$\times$ & 180.27$\times$ & 310.45$\times$ & 163.93$\times$ \\ \hline 
B\_cereus\_M& 2068138 & 477906536 & HOG & 0.003 & 8.654 & 4.070 & 0.698 & 0.922 \\ 
 & & & KMP & 0.00056 & 693.670 & 1420.341 & 664.633 & 668.875 \\ 
 & & & Rel & 0.19$\times$ & 80.16$\times$ & 348.98$\times$ & 952.20$\times$ & 725.46$\times$ \\ \hline 
B\_fragilis\_H& 12317015 & 1198165883 & HOG & 0.008 & 74.648 & 30.248 & 4.583 & 6.302 \\ 
 & & & KMP & 0.00042 & 1354.231 & 6524.023 & 1387.009 & 1392.511 \\ 
 & & & Rel & 0.05$\times$ & 18.14$\times$ & 215.68$\times$ & 302.64$\times$ & 220.96$\times$ \\ \hline 
M\_abscessus\_H& 5507950 & 492353747 & HOG & 0.002 & 28.556 & 11.384 & 1.627 & 2.376 \\ 
 & & & KMP & 0.00048 & 985.172 & 2904.419 & 932.243 & 934.678 \\ 
 & & & Rel & 0.24$\times$ & 34.50$\times$ & 255.13$\times$ & 572.98$\times$ & 393.38$\times$ \\ \hline 
M\_abscessus\_M& 1986962 & 332468005 & HOG & 0.002 & 7.945 & 4.026 & 0.656 & 0.883 \\ 
 & & & KMP & 0.00072 & 848.449 & 1735.039 & 845.944 & 817.385 \\ 
 & & & Rel & 0.36$\times$ & 106.79$\times$ & 430.96$\times$ & 1289.55$\times$ & 925.69$\times$ \\ \hline 
R\_sphaeroides\_H& 7727979 & 583472776 & HOG & 0.003 & 51.254 & 19.836 & 2.811 & 4.310 \\ 
 & & & KMP & 0.00060 & 1962.607 & 4959.768 & 1970.710 & 1973.718 \\ 
 & & & Rel & 0.20$\times$ & 38.29$\times$ & 250.04$\times$ & 701.07$\times$ & 457.94$\times$ \\ \hline 
R\_sphaeroides\_M& 1509743 & 208072134 & HOG & 0.002 & 5.954 & 2.925 & 0.516 & 0.722 \\ 
 & & & KMP & 0.00068 & 585.568 & 1159.319 & 574.844 & 573.252 \\ 
 & & & Rel & 0.34$\times$ & 98.35$\times$ & 396.35$\times$ & 1114.04$\times$ & 793.98$\times$ \\ \hline 
S\_aureus\_H& 7601994 & 741898249 & HOG & 0.005 & 32.004 & 13.138 & 1.832 & 2.729 \\ 
 & & & KMP & 0.00038 & 855.693 & 3536.636 & 859.875 & 866.531 \\ 
 & & & Rel & 0.08$\times$ & 26.74$\times$ & 269.19$\times$ & 469.36$\times$ & 317.53$\times$ \\ \hline 
V\_cholerae\_H& 3854210 & 353704636 & HOG & 0.002 & 15.161 & 6.823 & 1.085 & 1.451 \\ 
 & & & KMP & 0.00045 & 641.944 & 2146.911 & 638.035 & 638.916 \\ 
 & & & Rel & 0.23$\times$ & 42.34$\times$ & 314.66$\times$ & 588.05$\times$ & 440.33$\times$ \\ \hline 
V\_cholerae\_M& 1577183 & 293246651 & HOG & 0.002 & 5.826 & 2.967 & 0.510 & 0.670 \\ 
 & & & KMP & 0.00064 & 607.380 & 1142.848 & 579.774 & 579.888 \\ 
 & & & Rel & 0.32$\times$ & 104.25$\times$ & 385.19$\times$ & 1136.81$\times$ & 865.50$\times$ \\ \hline 
X\_axonopodis\_H& 11669945 & 1096782704 & HOG & 0.009 & 76.691 & 30.474 & 5.002 & 7.976 \\ 
 & & & KMP & 0.00046 & 1591.124 & 6404.583 & 1632.711 & 1626.796 \\ 
 & & & Rel & 0.05$\times$ & 20.75$\times$ & 210.17$\times$ & 326.41$\times$ & 203.96$\times$ \\ \hline 
\end{tabular}
\vspace{1em}
\caption{Running time (in $ms$) of queries for genome assemblies datasets for HOG and KMP based algorithms.} 
\vspace{-2em}
\label{tab:realQueriesKMP} 
\end{table}

\section{Pseudocodes for APSP variants}
\label{apn:pseudocodes}

The algorithm for Paired string query ($OneToOne$) is presented in \Cref{alg:onetoone}, while that of Multi String Query $OneToAll$ is shown in \Cref{alg:onetoall}. This algorithm can be extended to the other Multi String Queries as shown in \Cref{algo:onetoallEx}.

\begin{algorithm}
\BlankLine
$v \gets p_i$\; 
\While{\textbf{true}}{
$(min,max)
\gets$ $SubTreeMinMax(v)$\;
\lIf{$j\geq min$ $\textbf{and}$ $j\leq max$}
{\textbf{return} $v$}
\lElse{$v \gets$ Suffix link of $v$}
}

\caption{$OneToOne(i,j)$ 
}
\label{alg:onetoone}
\end{algorithm}

\begin{algorithm}
\ForEach{$j\in [1,k]$}{$ans[j]\gets \epsilon$\; $next[j]\gets j$\;}
\BlankLine
$v \gets p_i$\;
\While{$v \neq \epsilon$}{
$(min,max)
\gets SubTreeMinMax(v)$\; 
$j \gets min$\; 
\While{$j \leq max$}{
\lWhile{$next[j]\neq j$ \textbf{ and } $next[j] \leq max$}{ $j\gets next[j]$}

\If{$next[j]==j$}{
$ans[j] \gets v$\;
$j \gets j+1$\;
}
}
$next[min] \gets max+1$\;
$v \gets$ suffix link of $v$\;
}

\textbf{return} $ans$\;



\caption{$OneToAll(i)$ 
}
\label{alg:onetoall}
\end{algorithm}

\begin{algorithm}
\ForEach{$j\in [1,k]$}
{
{\color{blue} $ans_o[j]\gets \epsilon$}\;
$next[j]\gets j$\;}
{\color{red} $ans_r\gets \emptyset$} $|$
{\color{cyan} $ans_c\gets 0$} $|$
{\color{magenta} $ans_t\gets \emptyset$}
\; 

\BlankLine
$v \gets p_i$\;
\While{$v \neq \epsilon$ 
{\color{red} \textbf{and} $|v|\geq l$} $|$ 
{\color{cyan} \textbf{and} $|v|\geq l$} $|$ {\color{magenta} \textbf{and} $|ans_t|\leq c$}}{
$(min,max)
\gets SubTreeMinMax(v)$\; 
$j \gets min$\; 
\While{$j \leq max$}{
\lWhile{$next[j]\neq j$ \textbf{ and } $next[j] \leq max$}{ $j\gets next[j]$}

\If{$next[j]==j$}{

{\color{blue} $ans_o[j] \gets v$} $|$
{\color{red} $ans_r \gets ans_r\cup \{j\}$} $|$
{\color{cyan} $ans_c \gets ans_c+1$} $|$
{\color{magenta} $ans_t \gets ans_t\cup \{j\}$}\;
$j \gets j+1$\;
}
}
$next[min] \gets max+1$\;
$v \gets$ suffix link of $v$\;
}

\textbf{return} $ans$\; 

\caption{Multi String Query: {\color{blue}$OneToAll(i)$} $|$ {\color{red}$Report(i,l)$} $|$
{\color{cyan}$Count(i,l)$} $|$
{\color{magenta}$Top(i,c)$} 
}
\label{algo:onetoallEx}
\end{algorithm}

\end{document}